\documentclass[12pt]{amsart}

\usepackage{hyperref}

\numberwithin{equation}{section}

\usepackage{amsmath}
\usepackage{amsfonts}
\usepackage{amsthm}
\usepackage[english]{babel}
\usepackage{graphicx}
\usepackage[all]{xy}

\newtheorem{theorem}{Theorem}[section]
\newtheorem{lemma}{Lemma}[section]
\newtheorem{definition}{Definition}[section]
\newtheorem{corollary}{Corollary}[section]

\newtheorem{remark}{Remark}[section]
\newtheorem{example}{Example}[section]

\newcommand{\mR}{\mathbb{R}}
\newcommand{\mC}{\mathbb{C}}

\newcommand{\mZ}{\mathbb{Z}}

\newcommand{\mK}{\mathbb{K}}

\newcommand{\cM}{\mathcal{M}}

\newcommand{\cC}{\mathcal{C}}

\newcommand{\cK}{\mathcal{K}}

\newcommand{\cG}{\mathcal{G}}
\newcommand{\cW}{\mathcal{W}}
\newcommand{\cO}{\mathcal{O}}
\newcommand{\cV}{\mathcal{V}}
\newcommand{\cN}{\mathcal{N}}
\newcommand{\cU}{\mathcal{U}}
\newcommand{\cA}{\mathcal{A}}
\newcommand{\cB}{\mathcal{B}}

\newcommand{\nd}{{\rm and }}

\newcommand{\osp}{\mathfrak{osp}(m|2n)}

\newcommand{\End}{{\rm{End}}}
\newcommand{\Hom}{{\rm{Hom}}}

\begin{document}
\title[Invariant integration on supergroups]
{Invariant integration on orthosymplectic and unitary supergroups}

\author{K. Coulembier}

\address{Department of Mathematical Analysis, Ghent University,  Krijgslaan 281, 9000 Gent,
Belgium}
\thanks{The first author is a Ph.D. Fellow of the Research Foundation - Flanders (FWO)}
\email{Coulembier@cage.ugent.be}

\author{R.B.  Zhang}
\address{School of Mathematics and Statistics,
University of Sydney, Sydney, Australia}
\email{ruibin.zhang@sydney.edu.au}

\begin{abstract}
The orthosymplectic supergroup $OSp(m|2n)$ and unitary supergroup $U(p|q)$ are studied following a new approach that starts from Harish-Chandra pairs and links the sheaf-theoretical supermanifold approach of Berezin and others with the differential geometry approach of Rogers and others. The matrix elements of the fundamental representation of the Lie supergroup $\cG$ are expressed in terms of functions on the product supermanifold $\cG_0\otimes\mR^{0|N}$, with $\cG_0$ the underlying Lie group and $N$ the odd dimension of $\cG$. This product supermanifold is isomorphic to the supermanifold of $\cG$. This leads to a new expression for the standard generators of the corresponding Lie superalgebra $\mathfrak{g}$ as invariant derivations on $\cG$. Using these results a new and transparent formula for the invariant integrals on $OSp(m|2n)$ and $U(p|q)$ is obtained.
\end{abstract}

\maketitle


\tableofcontents

\section{Introduction}

Supersymmetry was first introduced into particle physics in the 70s. Since then it has had enormous impact on the development of theoretical high energy physics, particularly on the theories of supergravity and superstrings. Supersymmetry has also permeated many other areas of physics such as random matrix theory and condensed matter physics (see e.g. \cite{MR0708812, MR0820690, MR1421928}).

The algebraic foundation of supersymmetry lies in the theory of Lie supergroups. In this article we study the orthosymplectic supergroup $OSp(m|2n)$ and unitary supergroup $U(p|q)$, in particular, invariant integration on them. Invariant integration on Lie supergroups requires a super version of the unique Haar measure on locally compact Lie groups, which will be developed here within the context of the superalgebras of functions on supergroups with comultiplications (when restricted to the regular functions, they become Hopf superalgebras). More specifically, we wish to construct formulae of the form \eqref{GenexprInt}. An outline of our construction is given in Section \ref{overzicht}, which should be helpful for understanding the aims and main results of the paper without going through all the technical details. 

Our main results are explicit and transparent formulae for the invariant integrals (that is, generalisations of Haar measures) on $OSp(m|2n)$ and $U(p|q)$, which are respectively given in Theorem \ref{InvInt} and Theorem \ref{InvInt2}. The corresponding results, but written in the approach to supergroups of \cite{MR778559, MR1124825, MR0725288, MR574696}, are also presented in section \ref{sect:UOSp}. The new expressions for the invariant integrals enable us to prove their non-degeneracy (see Corollary \ref{Int-nondegeneracy} and Corollary \ref{Int-nondegeneracy2} for the precise statements), which is a question that could not be resolved previously.

The invariant integrals obtained in the present paper can be applied directly to random matrix theory. They can be used to extend results on integrals of monomials on classical Lie groups to supergroups (see e.g. \cite{MR2217291, MR2385262}), and to generalize Itzykson-Zuber or Harish-Chandra integrals to Lie supergroups (see e.g. \cite{GK, MR2081650, MR2492638, MR1421928}).  One simple application to random matrix theory is discussed Section \ref{sect:UOSp}, where we give an explicit construction of the invariant measure on $UOSp(m|2n)$ implicitly used in \cite{MR2081650}. It was demonstrated in \cite{MR1421928} that Harish-Chandra-Itzykson-Zuber integrals on supergroups naturally appeared in the theory of disordered metals within a nonlinear $\sigma$-model approach (see e.g. \cite{MR0708812}) to computing physical quantities at low wave lengths. It was the supersymmetric techniques that led to considerable progress in analytical computations in this area.

From a mathematical point of view, the expressions for the integral over supergroups in the present paper can help to develop Fourier theory (as in \cite{Hilgert}) on Lie supergroups. Note that the invariant integral for the abelian Lie supergroup $\mR^{p|q}$ corresponds to the ordinary Berezin integral \cite{Berezin, MR574696, MR1845224}. Finally the quantum and graded Haar measure is an important tool in the study of quantum and supergroups, see e.g. \cite{MR1462497}.


Our approach to invariant integrals on Lie supergroups is algebraical rather than measure theoretical as, for instance, in \cite{MR2667819}. This approach was started in \cite{MR1845224, MR2172158}. Earlier work on integrals on Hopf algebras can be found in e.g. \cite{MR1462497, MR0207883, MR0242840}. In \cite{MR1845224} the uniqueness of invariant integrals over Lie supergroups was proven. In \cite{MR2172158} the existence of this invariant integral was proven for all complex finite dimensional classical Lie superalgebras. The Lie superalgebra $\mathfrak{osp}(m|2n)$ of the orthosymplectic supergroup and the Lie superalgebra $\mathfrak{u}(p|q)$ of the unitary supergroup are real forms of such classical Lie superalgebras. An explicit formula for the invariant integral was presented in \cite{MR2172158}. However, the resulting expression is very unwieldy in general; only in the cases $OSp(1|2n)$, $OSp(2|2n)$ and $OSp(3|2)$ the formula \cite{MR1845224} becomes readily usable for practical computations. In the present paper a new approach to $OSp(m|2n)$ and $U(p|q)$ will be adopted, which will produce a transparent expression for the unique invariant integrals on these Lie supergroups for all values of $(m,n)$ and $(p,q)$ respectively.

There are different approaches to Lie supergroups. One approach starts from the sheaf-theoretical formulation of supermanifolds (\cite{Berezin, MR0580292, MR0760837}) and defines Lie supergroups as supermanifolds equipped with super multiplications. This is the approach that will be followed in this article, which implies that a supermanifold is by definition an object of this category. Although this approach is elegant, explicit examples of Lie supergroups are more easily obtained in other approaches. In the approach to supermanifolds taken in \cite{MR778559, MR574696}, a Lie supergroup is defined as a matrix group of supermatrices with their entries belonging to a certain Grassmann algebra $\Lambda_Q$ of degree $Q$ with $Q$ `big enough', see e.g. \cite{MR1124825, MR0725288}. These two approaches are known to be equivalent \cite{MR0554332}. Another way to explicitly construct Lie supergroups makes use of the equivalence between the categories of Lie supergroups and of super Harish-Chandra pairs \cite{MR2555976, MR0760837}, where a super Harish-Chandra pair consists of a Lie superalgebra and a Lie group with certain compatibility conditions.

In this paper we define the orthosymplectic supergroup $OSp(m|2n)$ and the unitary supergroup $U(p|q)$ as the unique Lie supergroups corresponding to the Harish-Chandra pairs $(O(m)\times Sp(2n),\osp)$ and $(U(p)\times U(q),\mathfrak{u}(p|q))$. We then construct the functor from Harish-Chandra pairs to Lie supergroups, $(\mathfrak{g},\cG_0)\to\cG$, in a new way. In order to do this, we consider the fundamental representation of the Lie supergroup. This corresponds to the defining representation of the Harish-Chandra pair.  We express the matrix elements of the representation in terms of functions on the underlying Lie group taking values in the algebra of $N$ free Grassmann variables (with $N$ the odd dimension of the Lie supergroup). This links the first approach to Lie supergroups with the second, and in some sense, `solves' the second approach by introducing an expression for the matrix entries in terms of the Grassmann variables. This shows that $Q$ `big enough' does in fact mean $Q\ge N$. We also very briefly discuss the interpretation of our formulas for the second approach to supermanifolds. With a minor adjustment we obtain formulas for the invariant measure discussed in e.g. \cite{MR2081650}.

The link between the matrix elements of the fundamental representation and the functions on the Lie supergroup introduces the multiplication on the supermanifold, leading to the structure of a Lie supergroup. Then, we can use the well-known functor from Lie supergroups to Harish-Chandra pairs to derive an action of the Harish-Chandra pair of a Lie supergroup ${\mathcal G}$ on the underlying supermanifold of ${\mathcal G}$. In particular, the Lie superalgebra $\mathfrak g$ acts by derivations.

Inputting this fact into the characterization of invariant integration used on $OSp(m|2n)$ in \cite{CouOSp}, one obtains a differential equation in the Grassmann variables, the solution of which gives rise to a very explicit and transparent formula for the integral on $OSp(m|2n)$. In the cases $OSp(m|2)$ and $OSp(1|2n)$, the result is particularly elegant. Results for the Lie supergroup $U(p|q)$ immediately lead to an expression for the invariant integral on $U(p|q)$, which is very elegant for all values of $p$ and $q$. The new expressions for the integrals on $OSp(m|2n)$ and $U(p|q)$ also enable us to prove their non-degeneracy. Our formula also leads to 
an explicit construction of the invariant measure on $UOSp(m|2n)$ implicitly used in \cite{MR2081650}, see Section \ref{sect:UOSp}.

The paper is organized as follows. First, a short introduction into the theory of Lie supergroups and Lie superalgebras is given. This is followed by a general outline of how we will construct Lie supergroups corresponding to Harish-Chandra pairs. The orthosymplectic superalgebra $\osp$ and unitary superalgebra $\mathfrak{u}(p|q)$ are introduced using their corresponding fundamental representations. These representations are used to obtain an expression for the multiplication of the corresponding Lie supergroups $OSp(m|2n)$ and $U(p|q)$. Then, an expression for the generators of the Lie superalgebras $\osp$ and $\mathfrak{u}(p|q)$ as derivatives on the Lie supergroups, is calculated. Finally, using the obtained results, the expression for the invariant integral on $OSp(m|2n)$ and $U(p|q)$ is calculated. The case $OSp(m|2n)$ needs some extra attention, since $OSp(m|2n;\mC)$ does not have a representable compact real form, see discussion in \cite{MR2276521, MR0714220}. The simplest case $p=1=q$ is used to demonstrate the calculation of integrals of monomials over $U(1|1)$. Then we extend the results to the physical approach to supergroups. Finally a short conclusion gives an overview of the most important results.

\section{Preliminaries}
\label{preliminaries}

We give a short introduction to superspaces, superalgebras, supermanifolds and Lie supergroups. The material can be found in,  e.g., \cite{Berezin, MR2069561}. For a field $\mK$, we introduce the super vector space (i.e., $\mZ_2$-graded vector space) $\cV=\cV_0\oplus\cV_1=\mK^{k|l}$ with the degree $0$ subspace $\cV_0=\mK^k$ and degree $1$ subspace $\cV_1=\mK^l$.  Call
$\cV_0$ and $\cV_1$ the even and odd subspaces of $\cV$ respectively. The standard basis of $\mK^{k|l}$ consists of the vectors $e_j$
for $1\le j\le k+l$, where
\begin{equation}
\label{basisvec}
e_j=(0,\cdots,0,1,0,\cdots,0)\quad\mbox{with $1$ at the $j$-th position}.
\end{equation}
The elements $e_j$ with $1\le j\le k$ span $\cV_0$, and $e_j$ with $k+1\le j\le k+l$ span $\cV_1$. A vector $u$ belonging to $\cV_0\cup\cV_1$ is called homogeneous, and in this case we define $|u|=\alpha$ for $u\in\cV_\alpha$ where $\alpha\in\mZ_2=\mZ/(2\mZ)$.
We also introduce a function
\begin{eqnarray*}
[\cdot]:\{1,2,\cdots, k+l\}\to\mZ_2, \quad
\text{$[j]=0$ if $j\le k$ and $[j]=1$ otherwise}.
\end{eqnarray*}
Then $|e_j|=[j]$ for all $j$.

A superalgebra $\cA$ over $\mK$ is a super vector space equipped with a multiplication $m:\cA\otimes \cA\to\cA$ that preserves the gradation in the sense that
$
|m(a\otimes b)|=|a|+|b|
$
for $a$ and $b$ homogeneous. The sum $|a|+|b|$ is understood to be (mod $2$) since it is inside $\mZ_2$. Hereafter definitions for various algebraic structures will often be given for homogeneous elements only and assumed to be defined for the entire super vector space through linearity.

A super $\mK$-bialgebra is a superalgebra $\cA$ with a comultiplication $\Delta:\cA\to \cA\otimes\cA$, which is an algebra homomorphism and preserves the gradation. Here the multiplication on $\cA\otimes\cA$ for $\cA$ being a superalgebra is defined by
\begin{eqnarray*}
(a\otimes b)(c\otimes d)&=&(-1)^{|b||c|}(ac\otimes bd).
\end{eqnarray*}
The fact that the comultiplication is an algebra homomorphism can then be expressed as $\Delta(ab)=\Delta(a)\Delta(b)$ for $a,b\in\cA$.

The algebra of regular functions on a Lie group becomes a bialgebra with comultiplication given by the pullback of the multiplication, $\Delta(f)(A\otimes B)=f(A\cdot B)$ for $f\in\cC^\infty(\cG_0)$ and $A,B\in\cG_0$ for a Lie group $\cG_0$. For a matrix Lie group the matrix elements $x_{ij}\in\cC^\infty(\cG_0)$ (defined by $x_{ij}(A)=A_{ij}$ for $A\in\cG_0$) clearly satisfy $\Delta(x_{ij})=\sum_{k}x_{ik}\otimes x_{kj}$.

\begin{definition}
\label{defHopfrep}
A corepresentation $\chi$ of a super bialgebra $\cA$ with comultiplication $\Delta$ on a super vector space $\cV$ is a gradation preserving linear mapping
\begin{eqnarray*}
\chi&:&\cV\to \cA\otimes\cV
\end{eqnarray*}
satisfying $(\epsilon\otimes id_{\cV})\circ\chi=id_{\cV}$ with $\epsilon$ the counit and $(\tau\otimes id_{\cV})\circ(id_\cA\otimes \chi)\circ \chi=(\Delta\otimes id_\cV)\circ\chi$, where the flip operator is defined by $\tau(a\otimes b)=(-1)^{|a||b|}b\otimes a$ for all homogeneous $a,b\in\cA$.
\end{definition}

For any two super vector spaces $\cV=\cV_0\oplus \cV_1$ and $\cW=\cW_0\oplus \cW_1$, we denote by $\Hom(\cV, \cW)$ the vector space of $\mC$-linear maps from $\cV$ to $\cW$. It is naturally $\mZ_2$-graded, thus is a super vector space,  with
\[
\begin{aligned}
\Hom(\cV, \cW)_0=\Hom(\cV_0, \cW_0)\oplus\Hom(\cV_1, \cW_1), \\
\Hom(\cV, \cW)_1=\Hom(\cV_0, \cW_1)\oplus\Hom(\cV_1, \cW_0).
\end{aligned}
\]
We denote $\Hom(\cV, \cV)$ by $\End(\cV)$, which forms an associative superalgebra under the composition of linear maps.

The super vector space $\End(\mK^{k|l})$ can be identified with the super vector space $\mK^{(k+l)\times(k+l)}$ of matrices.
It has a homogeneous basis consisting of elements $E_{ab}$ given by
\begin{eqnarray*}
E_{ab}\cdot e_c&=&\delta_{bc}e_a,
\end{eqnarray*}
with $\{e_c\mid c=1,\cdots,k+l\}$ the standard basis for $\mK^{k|l}$. The ${\mathbb Z}_2$-gradation is given by $|E_{ab}|=[a]+[b]$. 

Recall that a Lie superalgebra is a superalgebra where the multiplication satisfies a graded version of the Jacobi identity and is super anti-symmetric, see e.g. \cite{MR051963, MR0760837}. Now $\End(\mK^{k|l})$ is a Lie superalgebra with Lie superbracket given by the graded commutator,
\begin{eqnarray}
\label{supercomm}
[X,Y]&=&XY-(-1)^{|X||Y|}YX
\end{eqnarray}
for homogeneous elements $X,Y\in$ $\End(\mK^{k|l})$. This is the general linear Lie superalgebra $\mathfrak{gl}(k|l; \mK)$. We denote the natural representation of $\mathfrak{gl}(k|l;\mK)$ on $\mK^{k|l}$ by $\rho^\pi$. For $X\in \mathfrak{gl}(k|l; \mK)$, we have the
$(k+l)\times(k+l)$-matrix $(X_{a b})$ defined by
\begin{eqnarray}
\label{fundrepgl}
\rho^\pi(X)\cdot e_c&=&\sum_{a=1}^{k+l}X_{ac}e_a.
\end{eqnarray}

The standard definition of a supermanifold (see e.g. \cite{MR0554332, Berezin, MR0580292, MR0760837, MR1124825}) is of a sheaf theoretical nature. A sheaf $\cO$ of superalgebras with unity on an $m$-dimensional manifold $\cM_0$ consists of a map from each open subset $U$ of $\cM_0$ into a superalgebra $\cO(U)$ and a collection of superalgebra homomorphisms $\cO(U)\to\cO(V)$ for each pair of open sets $U$ and $V$ of $\cM_0$ such that $V\subset U$.  These superalgebra homomorphisms, called the restriction maps, satisfy a series of axioms, see, e.g., \cite{MR0580292}. In particular, given a sheaf of superalgebras, one automatically has a sheaf of superalgebras $\cO^{V}$ on each open subset $V$ of $\cM_0$ with $\cO^{V}(U)=\cO(U)$ for each open subset $U$ of $V$. The standard example is the sheaf of smooth functions on a manifold $\cM_0$, denoted by $\cC^\infty_{\cM_0}$. We will also use the notation $\cC^\infty(\cM_0)=\cC^\infty_{\cM_0}(\cM_0)$ for the commutative algebra of smooth functions on the manifold. For a continuous map of manifolds $f:\cM_0\to \cN_0$ and a sheaf $\cO$ on $\cM_0$ the push forward of $\cO$ is a sheaf on $\cN_0$ defined by $f_{\ast}\cO (U)=\cO (f^{-1}(U))$, for each open $U\subset\cN_0$.

\begin{definition}
\label{supmandef}
A supermanifold $\cM=(\cM_0,\cO_{\cM})$ of super dimension $D|N$ is a ringed space with $\cM_0$ a smooth $D$-dimensional manifold and $\cO_{\cM}$ a sheaf of $\mR$-superalgebras with unity on $\cM_0$, where $\cO_{\cM}$ is required to satisfy the local triviality condition that there exists an open cover $\{U_i\}_{i\in I}$ of $\cM_0$ and isomorphisms of sheaves of superalgebras $T_i:\cO_{\cM}^{U_i}\to \cC^\infty_{U_i}\otimes\Lambda_{N}$, where $\Lambda_{N}$ is the Grassmann algebra of degree $N$ (that is generated by $N$ anti-commuting variables).
\end{definition}
Here $\cO_{\cM}$ is the structure sheaf of $\cM$. Its sections (the elements of the superalgebra $\cO_{\cM}(\cM_0)$) are referred to as superfunctions on $\cM$.

\begin{remark}
The existence of $T_i:\cO_{\cM}^{U_i}\to \cC^\infty_{U_i}\otimes\Lambda_{N}$ corresponds to the expansion, in the physics literature, of superfunctions in the Grassmannian coordinates with coefficients being usual smooth functions on $U_i$.
\end{remark}

A morphism of supermanifolds $\Phi:\cM\to\cN$ is a morphism of ringed spaces $(\phi,\phi^\sharp)$. The mapping $\phi: \cM_0\to\cN_0$ is a manifold morphism and $\phi^\sharp:\cO_{\cN}\to\phi_{\ast}\cO_{\cM}$ is a morphism of sheaves local on each stalk. The morphism $\Phi$ is entirely determined \cite{MR0580292} by its induced mapping of sections
\begin{eqnarray*}
\phi^\sharp_{\cN_0}:\cO_\cN(\cN_0)\to\cO_{\cM}\left(\phi^{-1}(\cN_0)\right).
\end{eqnarray*}
To keep the notation simple, we will also denote $\phi^\sharp_{\cN_0}$ by $\phi^\sharp$.

Each point $p\in\cM_0$ defines a morphism $\delta_p:(\{\ast\},\mR)\to(\cM_0,\cO_{\cM}$), with $\{\ast\}$ the manifold consisting of one point. This is defined by
\begin{eqnarray}
\label{evalpunt}
\delta_p^{\sharp}(f)&=&[f]_0(p)
\end{eqnarray}
where $[\cdot]_0$ is the canonical projection $\cO_{\cM}(\cM_0)\to\cC^\infty_{\cM_0}(\cM_0)$, also denoted by $\delta^\sharp_\cM$. The identity morphism on a supermanifold $\cM$ is given by
\begin{eqnarray*}
id_\cM=(id_{\cM_0},id^\sharp_{\cM}):&&\cM\to\cM.
\end{eqnarray*}

We define the supermanifold morphisms $\rho:\cM\to\cM\otimes \cM$ by $\rho_0(u)=u\otimes u$ and $\rho^\sharp(f\otimes g)=fg$,  and $C:\cM\to(\ast,\mR)$ by $C^\sharp(\lambda)=\lambda 1_{\cM}$ with $1_{\cM}=1_{\cM_0}$ the unit function on $\cM_0$.
\begin{definition}
\label{defLiegr}
A Lie supergroup is a supermanifold $\cG=(\cG_0,\cO_\cG)$ equipped with the additional structure of a supermanifold morphism $\mu=(\mu_0,\mu^\sharp):\cG\otimes\cG\to\cG$, an involutive diffeomorphism $\nu=(\nu_0,\nu^\sharp):\cG\to\cG$ and a distinguished point $e_\cG\in\cG_0$. These satisfy
\begin{eqnarray*}
\mu\circ\left(id_{\cG}\otimes \mu\circ \left(id_\cG\otimes id_\cG\right)\right)&=&\mu\circ\left(\mu\circ\left( id_\cG\otimes id_\cG\right)\otimes id_{\cG}\right)\\
\mu\circ\left(id_\cG\otimes \delta_{e_\cG}\right)&=&id_\cG=\mu\circ\left(\delta_{e_\cG}\otimes id_\cG\right)\\
\mu\circ(id_\cG\otimes \nu)\circ\rho&=&\delta_{e_{\cG}}\circ C=\mu\circ( \nu\otimes id_\cG)\circ\rho.
\end{eqnarray*}
\end{definition}
With this definition, the underlying manifold $\cG_0$ of a Lie supergroup $\cG$ becomes a Lie group with multiplication $\mu_0$ and inversion $\nu_0$.

The structure sheaf of a Lie supergroup is always globally split, see e.g. \cite{MR0760837}. This means that for any Lie supergroup $\cG$ of dimension $D|N$, the structure sheaf satisfies
\begin{eqnarray}
\label{sheafsplit}
\cO_{\cG}&\cong &\cC^\infty_{\cG_0}\otimes \Lambda_N.
\end{eqnarray}
In other words, the anti-commuting variables correspond to global coordinates and that the properties of the sheaf $\cO_\cG$ correspond exactly to those of $\cC^\infty_{\cG_0}$.

For the definition of the multiplication $\mu$, the product supermanifold $\cG\otimes\cG$ was introduced. The corresponding sheaf $\cO_{\cG\otimes\cG}$ is a completion of the product of sheaves $\cO_{\cG}\otimes \cO_{\cG}$, see \cite{MR0094682}. Since the sheafs are globally split this completion is the same as that for ordinary Lie groups.
Restricting to the superalgebra $\cO_\cG(\cG_0)$ of regular functions on the supergroup (generated by the matrix elements of representations of $\cG$), we obtain a supercommutative super Hopf algebra (\cite{MR0594432, MR0242840}) with multiplication $\rho^\sharp$, comultiplication $\mu^\sharp$, antipode $\nu^\sharp$, unit $1_\cG$ and counit $\delta^\sharp_{e_\cG}$.

Let $Der\cA$ denote the super vector space of derivations of a super ($\mZ_2$-graded) algebra $\cA$. The homogeneous derivations are the homogeneous endomorphisms $X$ on $\cA$ which satisfy the graded Leibniz rule
\begin{eqnarray*}
X (ab)= X(a) b+ (-1)^{|X||a|}a X(b),\quad \mbox{ for }\, a,b\in\cA\quad \text{(with $a$ homogeneous)}.
\end{eqnarray*}
It is easily checked that $Der\cA$ is closed under the graded commutator \eqref{supercomm}. Therefore the super derivations form a Lie superalgebra.

\begin{definition}
\label{defLiealg1}
The Lie superalgebra $\mathfrak{g}$ corresponding to the Lie supergroup $\cG$ is the algebra with elements given by
\begin{eqnarray*}
\{X\in Der\cO_{\cG}(\cG_0)| \mu^\sharp \circ X &=&(X\otimes id_\cG^\sharp)\circ\mu^\sharp\}
\end{eqnarray*}
and with multiplication given by the supercommutator $[X,Y]=X\circ Y-(-1)^{|X||Y|}Y\circ X$.
\end{definition}

The distinguished point $e_{\cG}$ of $\cG_0$ will also be denoted by $0$. The evaluation at that point \eqref{evalpunt} therefore is denoted by $\delta_0^\sharp$. The following lemma is an immediate consequence of Definition \ref{defLiealg1}.
\begin{lemma}
\label{defLiealg2}
An $\mR$-vector space basis for the Lie superalgebra $\mathfrak{g}$ is given by the derivations
\[Z_j=\left(\delta^\sharp_{0}\circ Y_j\otimes id_{\cG}^\sharp\right)\circ\mu^\sharp\]
for $\{Y_j,j=1,\cdots,D\}$ the bosonic derivatives with respect to a set of local coordinates on $\cG_0$ in a neighborhood around the origin and $\{Y_j,j=D+1,\cdots,D+N\}$ the $N$ global Grassmann derivatives.
\end{lemma}

This lemma in particular enables us to obtain explicit realisations of the Lie superalgebra in terms of derivations on the superalgebra of functions on the supergroup, see, e.g., equations \eqref{defK} and \eqref{KonX} for the 
orthosymplectic supergroup. Also we immediately see from the lemma that the superdimensions of a Lie supergroup and its Lie superalgebra coincide.

In view of Definition \ref{defLiealg1}, we can define a mapping from Lie supergroups to pairs of Lie groups and Lie superalgebras, $\cG\to(\cG_0,\mathfrak{g})$. The Lie group and Lie superalgebra satisfy the compatibility conditions of a super Harish-Chandra pair:

\begin{definition}
\label{HaCh}
A Harish-Chandra pair is a pair $(\cG_0,\mathfrak{g})$, consisting of a  Lie group $\cG_0$ and a Lie superalgebra $\mathfrak{g}=\mathfrak{g}_0\oplus \mathfrak{g}_1$ with $\mathfrak{g}_0$ the Lie algebra of $\cG_0$ and a representation Ad of $\cG_0$ on $\mathfrak{g}$ such that
\begin{itemize}
\item Ad on $\mathfrak{g}_0$ is the usual adjoint action,
\item the differential of the action at the identity is equal to the Lie superbracket, restricted to $\mathfrak{g}_0\times\mathfrak{g}$.
\end{itemize}
\end{definition}

In order to complete the Harish-Chandra pair corresponding to $\cG$, the adjoint representation of $\cG_0$ on $\mathfrak{g}$ is defined as
\begin{eqnarray}
\label{Adj}
Ad(g)X&=&(\delta_{g^{-1}}^\sharp\otimes id_\cG^\sharp)\circ\mu^\sharp\circ X\circ (\delta_{g}^\sharp\times id_{\cG}^\sharp)\circ\mu^\sharp,
\end{eqnarray}
for $g\in\cG_0$ and $X\in\mathfrak{g}$. This action of $\cG_0$ on $\mathfrak{g}$ is canonically extended to the universal enveloping algebra $\cU(\mathfrak{g})$.  The functor $\cG\mapsto(\cG_0,\mathfrak{g})$ from the category of Lie supergroups to the category of Harish-Chandra pairs described above is an equivalence (see \cite{MR0580292}). The inverse functor is constructed in \cite{MR2555976, MR0760837}.

We denote by $Aut(V)$ the automorphism group on a classical vector space $V$.

\begin{definition}
\label{defrep2}
A representation of a super Harish-Chandra pair $(\cG_0,\mathfrak{g})$ on a graded vector space $\cV$ is a pair $\Pi=(\pi_0,\rho^\pi)$, with
\begin{itemize}
\item $\pi_0$: $\cG_0\to Aut(\cV_0)\times Aut(\cV_1)$ a group morphism
\item $\rho^\pi$: $\mathfrak{g}\to End(\cV)$ a Lie superalgebra morphism
\end{itemize}
such that $d\pi_0=\rho^\pi$ on $\mathfrak{g_0}$ and $\rho^\pi\left(Ad(g)X\right)=Ad(\pi_0(g))\rho^\pi(X)$ for $X\in\mathfrak{g}$ and $g\in\cG_0$.
\end{definition}
The adjoint representation of $Aut(\cV_0)\oplus Aut(\cV_1)$ on $End(\cV)$ is naturally defined.

A linear functional $T$ on $\cO_\cG(\cG_0)$ is called left-invariant if
\begin{eqnarray*}
(id_\cG^\sharp\otimes T)\circ \mu^\sharp&=&1_{\cG}T
\end{eqnarray*}
holds. This corresponds to the definition of invariant integration for Lie groups and to the notion of invariant integration on Hopf algebras.

The invariance of integration on a Lie supergroup can be expressed in terms of the Harish-Chandra pair, see \cite{CouOSp}:
\begin{lemma}
\label{invHopf}
Consider a Lie supergroup $\cG$ with multiplication $\mu$ and corresponding Lie superalgebra $\mathfrak{g}$. A linear functional $\int_{\cG}:\cO_\cG(\cG_0)\to\mR$ is left-invariant if and only if
\begin{itemize}
\item $(\delta_{g}^\sharp\times \int_{\cG})\circ \mu^\sharp =\int_{\cG}\quad$ for all $g\in\cG_0$,
\item $\int_{\cG} \circ X=0\qquad\quad$ for all $X\in \mathfrak{g}$.
\end{itemize}
\end{lemma}
A (left-)invariant functional on $\cO_{\cG}(\cG_0)$ is called an invariant integral on $\cG$.

In the present paper we aim to calculate explicit expressions for this integral which are useful to do explicit calculations. In \cite{MR1845224, MR2172158} an explicit construction was obtained which we briefly revue here. Denote by $\int_{\cG_0}$ the unique left-invariant integral on the Lie group $\cG_0$ and by
\begin{eqnarray}
\label{projbos}
\delta_G^\sharp&:&\cO_{\cG}\to\cC^\infty_{\cG_0}
\end{eqnarray}
the projection onto the underlying function by dropping the nilpotent part (setting the Grassmann variables equal to zero), also defined by $\delta_p^\sharp(f)=(\delta^\sharp_G(f))(p)$ using equation \eqref{evalpunt}. Lemma \ref{invHopf} then implies the following result:

\begin{corollary}
\label{ResSZ}
Consider a Lie supergroup $\cG$ with Lie superalgebra $\mathfrak{g}$. Assume there is some $Y\in \cU(\mathfrak{g})$ such that
\begin{itemize}
\item for all $X\in \mathfrak{g}$, $YX$ is an element of the right ideal $\mathfrak{g}_0\cU(\mathfrak{g})$,
\item for all $g\in\cG_0$, $Ad(g)(Y)-Y$ is an element of the right ideal $\mathfrak{g}_0\cU(\mathfrak{g})$.
\end{itemize}
Then the integral $\int_{\cG}=\int_{\cG_0}\circ\delta^\sharp_{\cG}\circ Y$ is left invariant.
\end{corollary}

In \cite{MR2172158} a method to to find the invariant element $Y\in\cU(\mathfrak{g})$ was given. But since it is too complicated for practical calculations we will develop a new expression for the integral on $\cG=OSp(m|2n)$ and $\cG=U(p|q)$. We will find explicit expressions for integrals of functions of the matrix elements $X_{ij}$. The formulas will be of the form
\begin{eqnarray*}
\int_{\cG}f(X)=\int_{\cG_0}d\nu(U)\int_{B,\theta}\alpha(U,\theta)f(X(U,\theta)),
\end{eqnarray*}
where $\nu$ represents the Haar measure on $\cG_0$, $U$ coordinates on $\cG_0$, $\int_B$ the Berezin integral on the Grassmann algebra, $\theta$ the Grassmann variables of the supermanifold $\cG$ and $X_{ij}(U,\theta)$ a mapping from coordinates on $\cG$ to the matrix elements.

From comparison with Corollary \ref{ResSZ} we find that our approach leads to a different way of calculating the operator $\delta^\sharp_{\cG}\circ Y$, as $\int_B\alpha$.

\section{Outline of the construction}
\label{overzicht}
The most transparent way to introduce a Lie supergroup is by making use of the equivalence of categories with Harish-Chandra pairs. The functor from the category of Lie supergroups (Definition \ref{defLiegr}) to that of super Harish-Chandra pairs (Definition \ref{HaCh}) was constructed in Section \ref{preliminaries}. The explicit construction of the inverse functor was introduced in \cite{MR0760837} and further studied in \cite{MR2555976}. However, in order to obtain the formula for invariant integration, in this paper we will introduce a different, more analytical construction of the Lie supergroup corresponding to a Harish-Chandra pair. An overview of that construction will be presented in this section. First we briefly sketch the construction of \cite{MR2555976, MR0760837}.

Starting from the Harish-Chandra pair $\hat{\cG}=(\cG_0,\mathfrak{g})$, the first step in the construction of the Lie supergroup $\cG$ is the definition of the supermanifold $\underline{\cG}=(\cG_0,\cO_\cG)$. The sheaf $\cO_\cG$ is defined, for each open $U\subset \cG_0$, by
\begin{eqnarray*}
\cO_\cG(U)&=&\mbox{Hom}_{\cU(\mathfrak{g}_0)}\left(\cU(\mathfrak{g}),\cC^\infty_{\cG_0}(U)\right),
\end{eqnarray*}
where the right hand side is the super vector space of $\cU(\mathfrak{g}_0)$-linear morphisms from $\cU(\mathfrak{g})$ to $\cC^\infty_{\cG_0}(U)$, which has a natural superalgebra structure induced by the co-superalgebra structure of $\cU(\mathfrak{g})$.
Here the $\mathfrak{g}_0$-action on $\cC^\infty_{\cG_0}(U)$ is the standard one, i.e., the differential of left translation. It is clear that
\begin{eqnarray*}
\cO_\cG(U)&\cong&\cC^\infty_{\cG_0}(U)\otimes \Lambda(\mathfrak{g}^\ast_1),
\end{eqnarray*}
as equation \eqref{sheafsplit} demands.

The multiplication $\mu$ and involution $\nu$ can then be naturally defined for the sheaf when expressed in the form  $\mbox{Hom}_{\cU(\mathfrak{g_0})}\left(\cU(\mathfrak{g}),\cC^\infty_{\cG_0}(U)\right)$, see \cite{MR2555976} for the explicit formulas. The supermanifold $\underline{\cG}$ equipped with the morphisms $\mu$ and $\nu$ then is a Lie supergroup $\cG$. The Harish-Chandra pair corresponding to this Lie supergroup (following the functor between categories introduced in Section \ref{preliminaries}) is the original $\hat{\cG}$. So this functor from Harish-Chandra pairs to Lie supergroups is in fact the inverse functor.

However, this does not define an expression for the comultiplication with respect to the $\cC^\infty_{\cG_0}\otimes \Lambda(\mathfrak{g}^\ast_1)$-expression of the sheaf. We will find an expression for (one choice of) the comultiplication on functions in $\cC^\infty(\cG_0)\otimes \Lambda(\mathfrak{g}^\ast_1)$, starting from $\hat{\cG}$, for $\cG$ equal to $OSp(m|2n)$ and $U(p|q)$. This will be approached in the following way.

The Harish-Chandra pairs to be considered can be defined through their natural representation $(\Pi,\cV)$ (see Definition \ref{defrep2}). This representation can be extended to $T(\cV)=\oplus_{k\ge 0} T(\cV)_k$, where $T(\cV)_k= \cV\otimes\cV\otimes\dots\otimes\cV$ ($k$ copies). The pairs we will consider correspond to sub-pairs $(\cG_0,\mathfrak{g})$ of $GL(k|l;\mK)=(GL(k;\mK)\times GL(l;\mK),\mathfrak{gl}(k|l;\mK))$ for certain values of $k$ and $l$ and $\mK$ equal to $\mR$ or $\mC$. This means the representation space is $\cV=\mK^{k|l}$. The groups of the pairs consist of the elements of $GL(k;\mK)\times GL(l;\mK)$ which leave a certain even linear functional $T: T(\cV)_2\to\mK$ invariant.
Here $T$ is required to be non-degenerate when regarded as a bilinear form on $\cV$. The superalgebras of the pairs consist of the elements of $Y\in\mathfrak{gl}(k|l;\mK)$ such that $T\circ Y=0$ on $T(\cV)_2$.

Then the representation space $\cV$ and the linear functional $T$ can be used to define a super bialgebra $\cA$. We start from an abstract super bialgebra $\cA$ with a corepresentation $\chi$ as in Definition \ref{defHopfrep}. Such a corepresentation defines elements $X_{ij}\in \cA$ (not necessarily independent) for $1\le i,j\le k+l$ by considering $\chi$ on the basis elements $e_j\in\cV$ in equation \eqref{basisvec},
\begin{eqnarray}
\label{corepOSp}
\chi(e_i)&=&\sum_{j=1}^{k+l}X_{ji}\otimes e_j.
\end{eqnarray}
For a vector $v=\sum_{i=1}^{k+l}v_ie_i$ we find $\chi(v)=\sum_{j=1}^{k+l}v'_j\otimes e_j$ with $v'_j=\sum_{j=1}^{k+l}X_{ji}v_i$ given by ordinary matrix multiplication. The gradation of $\cA$ on $X_{ij}$ is given by $|X_{ij}|=[i]+[j] $, since $\chi$ has to be even.

The corepresentation is canonically extended to the tensor space $T(\cV)$ by taking into account the gradation, e.g. $\chi: T(\cV)_2\to \cA\otimes T(\cV)_2$ is given by
\begin{eqnarray}
\label{extcorep}
\chi(e_i\otimes e_j)&=&\sum_{t,s=1}^{k+l}X_{ti}X_{sj}(-1)^{[t]([j]+[s])}\otimes e_t\otimes e_s.
\end{eqnarray}
Imposing the condition that the corepresentation fixes $T$,
\begin{eqnarray*}
(id_{\cA}\otimes T)\circ \chi&=&1_{\cA}\otimes T\qquad\mbox{on}\quad \mK^{k|l}\otimes\mK^{k|l},
\end{eqnarray*}
we arrive at relations among the generators $X_{ij}$. Therefore we can define the super bialgebra $\cA$ as the supercommutative super bialgebra generated by the elements $X_{ij}$, subject to these relations. The comultiplication needs to be defined as
\begin{eqnarray}
\label{defcomult}
\Delta (X_{ij})&=&\sum_{t=1}^{k+l}(-1)^{([i]+[t])([t]+[j])}X_{it}\otimes X_{tj}
\end{eqnarray}
in order for $\chi$ in equation \eqref{corepOSp} to be a proper corepresentation (Definition \ref{defHopfrep}). Because of the correspondence of this corepresentation of $\cA$ with the defining representation for the Harish-Chandra pair $\hat{\cG}$, the super bialgebra $\cA$ can be embedded into $\cC^\infty(\cG_0)\otimes \Lambda(\mathfrak{g}^\ast_1)$. By calculating such an embedding $\cA\subset \cC^\infty(\cG_0)\otimes \Lambda(\mathfrak{g}^\ast_1)$ we can extend the comultiplication $\Delta$ on the super bialgebra $\cA$ uniquely to a multiplication $\mu$ on the supermanifold $\underline{\cG}=(\cG_0,\cC^\infty_{\cG_0}\otimes\Lambda(\mathfrak{g}^\ast_1))$. This defines a Lie supergroup $\cG$ which has as Harish-Chandra pair, the original $\hat{\cG}=(\cG_0,\mathfrak{g})$. We have therefore constructed an alternative formulation of the inverse functor described above for $\cG$ equal to $OSp(m|2n)$ and $U(p|q)$.

Using these results we can calculate the invariant derivations on $\cC^\infty({\cG_0})\otimes\Lambda(\mathfrak{g}_1^\ast)$, which generate $\mathfrak{g}$ and by restriction of the multiplication an action of $\cG_0$ on $\cC^\infty({\cG_0})\otimes\Lambda(\mathfrak{g}_1^\ast)$.

\begin{remark}
Although not essential for what follows we remark that the action of $\mathfrak{g}$ and $\cG_0$ on $\cG$ as described above corresponds to an action of the Harish-Chandra pair $(\cG_0,\mathfrak{g})$ on the supermanifold $\cG$. The definition of such an action is given in Proposition 3.3 in \cite{MR2555976}.
\end{remark}

Schematically, the construction is given by
\[
\hat{\cG}=(\cG_0,\mathfrak{g}),\,\, (\Pi,\cV)\quad \to\quad \underline{\cG},\,\,\cA\subset\cC^\infty(\cG_0)\otimes \Lambda(\mathfrak{g}^\ast_1)\quad\to\quad \cG\quad \to\quad \mbox{action} \,\, \mbox{of}\,\, \hat\cG\,\, \mbox{on}\,\,\underline{\cG}.
\]
This action of $\hat{\cG}$ on $\underline{\cG}$ contains all the information and is a very elegant way to describe the Lie supergroup. In particular we will use it to construct the invariant integral on the Lie supergroup $\cG$ for $OSp(m|2n)$ and $U(p|q)$.

We should note that there are different possibilities for the embedding $\cA\subset\cC^\infty(\cG_0)\otimes \Lambda(\mathfrak{g}^\ast_1)$, which all lead to different multiplications $\mu$. The resulting supergroups are of course all isomorphic. Also the expression for the integral will differ for these choices, but the integral of matrix elements of representations of $\cG$ will still give the same result. In fact the resulting expressions for the integral will be of the form
\begin{eqnarray*}
\int_{\cG}\cdot=\int_{\cG_0}\int_{B}\alpha\cdot&:&\cC^\infty(\cG_0)\otimes \Lambda(\mathfrak{g}^\ast_1)\to\mR,
\end{eqnarray*}
with $\alpha\in \cC^\infty(\cG_0)\otimes\Lambda(\mathfrak{g}^\ast_1)$ and $\int_B$ the Berezin integral on $\Lambda(\mathfrak{g}^\ast_1)$, see \cite{Berezin}. The relevant integrals will be those of functions $f(X)$ where $X$ represent the elements of $\cA$ introduced above, which correspond to the matrix elements of the fundamental representation. The integral should therefore be calculated as
\begin{eqnarray}
\label{GenexprInt}
\int_{\cG}f(X)=\int_{\cG_0}d\nu(U)\int_{B,\theta}\alpha(U,\theta)f(X(U,\theta)),
\end{eqnarray}
where $\nu$ represents the Haar measure on $\cG_0$, $U$ coordinates on $\cG_0$, $\theta$ the Grassmann variables corresponding to $\Lambda(\mathfrak{g}^\ast_1)$ and $X_{ij}(U,\theta)\in\cC^\infty(\cG_0)\otimes\Lambda(\mathfrak{g}^\ast_1)$ the embedding $\cA\subset\cC^\infty(\cG_0)\otimes\Lambda(\mathfrak{g}^\ast_1)$. Each different choice of embedding $X(U,\theta)$ will lead to a different expression for the integral (a different $\alpha(U,\theta)$) such that the resulting expression $\int_{\cG}f(X)$ remains the same.

\section{The Lie superalgebras $\mathfrak{osp}(m|2n)$ and $\mathfrak{u}(p|q)$}
\label{defalg}

In this section we introduce the defining representations of the real Lie superalgebras $\mathfrak{osp}(m|2n)$ and $\mathfrak{u}(p|q)$.

The natural representation of $\mathfrak{gl}(k|l; \mK)$ on $\mK^{k|l}$ extends to a representation on the tensor space $T(\mK^{k|l})$ by the graded Leibniz rule
\begin{eqnarray*}
X\cdot \left(v_1\otimes v_2\otimes \cdots \otimes v_j\right)&=&\left(X\cdot v_1\right)\otimes v_2\otimes \cdots\otimes v_j\\
&+&(-1)^{|X||v_1|}v_{1}\otimes \left( X\cdot v_{2}\right)\otimes \cdots \otimes v_{j}\\
&+&\cdots+(-1)^{|X|(|v_1|+\cdots+|v_{j-1}|)}v_{1}\otimes\cdots\otimes v_{j-1}\otimes \left( X\cdot  v_{j}\right)
\end{eqnarray*}
for $X\in\mathfrak{gl}(k|l;\mK)$ and $v_i\in\mK^{k|l}$ homogeneous.

\begin{definition}
\label{defalgosp}
The Lie superalgebra $\osp$ is the subalgebra of $\mathfrak{gl}(m|2n;\mR)$ consisting of
\begin{eqnarray*}
\{Y\in \mathfrak{gl}(m|2n;\mR)| \, T(Y \cdot (v\otimes w)) =0\,\,\forall \,v,w\in\mR^{m|2n}\},
\end{eqnarray*}
with $T:\mR^{m|2n}\otimes\mR^{m|2n}\to\mR$ bilinear and satisfying $T(e_i\otimes e_j)=g_{ji}$ for $\{e_j\}$ the standard basis on $\mR^{m|2n}$. The orthosymplectic metric $g\in\mR^{(m+2n)\times (m+2n)}$ is given here by
\begin{eqnarray*}
g=\left( \begin{array}{cc} I_m&0\\ \vspace{-3.5mm} \\0&J
\end{array}
 \right)&\mbox{with}&J=\left( \begin{array}{cc} 0&I_{n}\\  \vspace{-3.5mm} \\-I_n&0
\end{array}
 \right).
 \end{eqnarray*}
\end{definition}
The algebra $\osp$ is generated by $\cK_{ij}\in \mathfrak{gl}(m|2n;\mR)$ for $1\le i\le j\le m+2n$,
\begin{eqnarray}
\label{defgenosp}
\cK_{ij}e_\alpha&=&g_{\alpha j}e_i-(-1)^{[i][j]}g_{\alpha i} e_j,
\end{eqnarray}
with $\{e_i\}$ defined in equation \eqref{basisvec}. These generators satisfy
\begin{eqnarray}
\nonumber
[\cK_{ij},\cK_{kl}]&=&\cK_{ij}\cK_{kl}-(-1)^{([i]+[j])([k]+[l])}\cK_{kl}\cK_{ij}\\
\label{stanosp}
&=&g_{kj}\cK_{il}+(-1)^{[i]([j]+[k])}g_{li}\cK_{jk}-(-1)^{[k][l]}g_{lj}\cK_{ik}-(-1)^{[i][j]}g_{ki}\cK_{jl}.
\end{eqnarray}
Then $\osp$ consists of matrices in $\mR^{(m+2n)\times (m+2n)}$ of the form
\begin{eqnarray*}
\left( \begin{array}{cc} A&C\\ \vspace{-3.5mm} \\JC^T&B
\end{array} \right)
\end{eqnarray*}
with $A^T=-A$ and $B^T=JBJ$. The super dimension of the Lie superalgebra $\osp$ as a super $\mR$-vector space is therefore $\frac{1}{2}m(m-1)+n(2n+1)|2mn$.

There is a well established notion of unitary representations of $\mC$-Lie superalgebras, see e.g. \cite{MR0424886, MR1075732}.
Here we reformulate it in terms of super Hermitian forms defined in the following way \cite{MR2207328}.
Consider the complex vector space $\mC^{p+q}$ with the standard hermitian inner product $\langle\cdot|\cdot\rangle$. Then using the gradation the vector space becomes the super vector space $\mC^{p|q}$ on which we define the super hermitian form $(\cdot|\cdot)$,
\begin{eqnarray*}
(u|v)&=& i^{|u||v|}\langle u|v\rangle
\end{eqnarray*}
for $u$ and $v$ homogeneous and extended to general $u$ and $v$ by linearity. This implies that $(u|v)=(-1)^{|u||v|}\overline{(v|u)}$ holds, which justifies the term super hermitian form.
\begin{definition}
For $T\in$ $End(\mC^{p|q})$, the super adjoint $T^\ast\in$ $\End(\mC^{p|q})$ is defined by the relation
\begin{eqnarray*}
(Tu|v)&=&(-1)^{|T||u|}(u|T^\ast v).
\end{eqnarray*}
\end{definition}
A straightforward calculation shows that $T^\ast=i^{|T|}T^\dagger$ holds with $T^\dagger$ the standard adjoint with respect to $\langle\cdot|\cdot\rangle$. The behavior with respect to the supercommutator \eqref{supercomm} is given by
\begin{eqnarray*}
[T,S]^\ast&=&-[T^\ast,S^\ast].
\end{eqnarray*}

A representation of a Lie superalgebra $\mathfrak{g}$ on $\mC^{p|q}$, $\lambda:\mathfrak{g}\to$ $\End(\mC^{p|q})$, is unitary if for each $X\in\mathfrak{g}$, $\lambda(X)^\ast=-\lambda(X)$ holds. This naturally leads to the definition of the unitary Lie  superalgebra $\mathfrak{u}(p|q)$.

\begin{definition}
The Lie superalgebra $\mathfrak{u}(p|q)$ is the $\mR$-Lie superalgebra consisting of
\begin{eqnarray*}
\{X\in \mathfrak{gl}(p|q;\mC)| \, X^\ast=-X\}
\end{eqnarray*}
and with product inherited from $\mathfrak{gl}(p|q;\mC)$.
\end{definition}

Thus $\mathfrak{u}(p|q)$ consists of matrices in $\mC^{(p+q)\times (p+q)}$ of the form
\begin{eqnarray*}
\left( \begin{array}{cc} A&C\\  \vspace{-3.5mm} \\-iC^\dagger&B
\end{array} \right)
\end{eqnarray*}
with $A^\dagger=-A$ and $B^\dagger=-B$. The super dimension of the Lie superalgebra $\mathfrak{u}(p|q)$ as a super $\mR$-vector space is hence $p^2+q^2|2pq$.

There is a slightly different characterization of the unitary Lie superalgebra, which will be useful later. To describe it, we introduce the complex conjugate $\mathfrak{gl}(p|q;\mC)$-representation $\overline{\rho^\pi}$ on $\mC^{p|q}$ of $\rho^\pi$ in equation \eqref{fundrepgl},
\begin{eqnarray*}
\overline{\rho^\pi}(Y)\cdot e_j&=&\sum_{k=1}^{p+q}\overline{Y}_{kj}e_k.
\end{eqnarray*}

Using this definition, the following lemma can be proved by a straightforward calculation.
\begin{lemma}
\label{metru}
The Lie superalgebra $\mathfrak{u}(p|q)$ consists of
\begin{eqnarray*}
\{Y\in \mathfrak{gl}(p|q;\mC)| \, L((\overline{\rho^\pi}\otimes \rho^\pi)(Y)\cdot (v\otimes w)) =0\},
\end{eqnarray*}
with $L:\mC^{p|q}\otimes \mC^{p|q}\to\mC$ given by $L(e_i\otimes e_j)=h_{ij}$ and bilinear. The matrix $h$ is given by $h=\left( \begin{array}{cc} I_p&0\\  \vspace{-3.5mm} \\0&i I_q
\end{array}
 \right)$.
\end{lemma}

\begin{remark}
There exists another compact real form of $\mathfrak{gl}(p;q|\mC)$ which corresponds to the elements of $\mathfrak{gl}(p;q|\mC)$ which satisfy $Y^\dagger=-i^{|Y|}Y$ instead of the choice $Y^\dagger=-(-i)^{|Y|}Y$ made here. It is straightfoward to extend all subsequent results on $U(p|q)$ from our choice to the other definition.
\end{remark}

\section{The Lie supergroup $OSp(m|2n)$}
\subsection{Definition}
\label{Definition1}
The Lie supergroup $OSp(m|2n)$ corresponds to the Harish-Chandra pair $(O(m)\times Sp(2n), \osp)$. The adjoint representation of $O(m)\times Sp(2n)$ on $\mathfrak{osp}(m|2n)$ is implied by considering the representation of the Harish-Chandra pair on $\mR^{m|2n}$, which embeds the pair into $(GL(m;\mR)\times GL(2m;\mR),\mathfrak{gl}(m|2n;\mR))$. We always use the notation $Sp(2n)$ for $Sp(2n;\mR)$, which is noncompact. The compact real form of $Sp(2n;\mC)$ will be denoted by $USp(2n)=Sp(2n;\mC)\cap U(2n)$.

The pair $(O(m)\times Sp(2n), \osp)$ corresponds to
\[(\cG_0,\mathfrak{g})\subset(GL(m;\mR)\times GL(2m;\mR),\mathfrak{gl}(m|2n;\mR))=(\mbox{Aut}(\mR^{m})\times\mbox{Aut}(\mR^{2n}),\,\mbox{End}(\mR^{m|2n}))\]
acting on $\mR^{m|2n}$ such that
\begin{eqnarray*}
T\circ S =T \quad \mbox{on} \quad \mR^{m|2n}\otimes \mR^{m|2n} & & \forall \,S\in \cG_0,\\
T\circ Y=0\quad \mbox{on} \quad \mR^{m|2n}\otimes \mR^{m|2n} & &\forall\, Y\in \mathfrak{g}.
\end{eqnarray*}
with $T$ given in Definition \ref{defalgosp}. From this definition it is clear that the adjoint action is defined as $Ad(S)Y=S^{-1}\circ Y\circ S$.

Now as described in Section \ref{overzicht}, we consider the supercommutative super bialgebra $\cA$, generated by abstract graded elements $X_{ij}$, $i,j=1,\cdots,m+2n$ with comultiplication \eqref{defcomult} and with corepresentation \eqref{corepOSp} on $\mR^{m|2n}$. For this corepresentation to be compatible with the representation of $(O(m)\times Sp(2n),\osp)$ described above, the corepresentation $\chi$ on $\mR^{m|2n}\otimes\mR^{m|2n}$ should satisfy
\begin{eqnarray*}
(id_{\cA}\otimes T)\circ \chi&=&T.
\end{eqnarray*}
Equation \eqref{extcorep} therefore implies that the relation
\begin{eqnarray}
\label{relforX}
\sum_{i,j=1}^{m+2n}(-1)^{[l]([l]+[j])}X_{ik}g_{ij}X_{jl}&=&g_{kl}
\end{eqnarray}
must hold for the abstract generators $X_{ij}$ of $\cA$.

\begin{remark}
This corresponds to a defining relation of a Lie supergroup within the approach of \cite{MR1124825, MR0725288}. There the matrix entries of the elements of the supergroup satisfy similar equations. Theorem \ref{Xifo} below describes the matrix elements in terms of free Grassman variables.
\end{remark}

Using the well-known super transpose of a super matrix $(X^T)_{ki}=(-1)^{[i]([i]+[k])}X_{ik}$ and noting that the metric $g$
only has a $(I, I)$ block (which is symmetric) and $(II, II)$ block (which skew symmetric), we can rewrite equation \eqref{relforX} as $\sum_{i,j}X_{jl}(X^T)_{ki}g_{ij}=g_{kl}$. In block matrix form this relation is given by
\begin{eqnarray*}
\left( \begin{array}{cc} X_{I,I}^T&X_{II,I}^T\\  \vspace{-3.5mm} \\X_{I,II}^T&X_{II,II}^T
\end{array}
 \right)
g\left( \begin{array}{cc} X_{I,I}&-X_{I,II}\\  \vspace{-3.5mm} \\X_{II,I}&X_{II,II}
\end{array}
 \right)
&=&g,
\end{eqnarray*}
where $X_{\alpha,\beta}^T$ is the classical transpose $(X_{\alpha,\beta})^T$ for $\alpha,\beta=I,II$.

This leads to 3 independent relations for the submatrices:
\begin{eqnarray}
\nonumber
X_{I,I}^TX_{I,I}+X_{II,I}^TJX_{II,I}&=&I_m,\\
\label{defrelations}
-X_{I,I}^TX_{I,II}+X_{II,I}^TJX_{II,II}&=&0,\\
\nonumber
-X_{I,II}^TX_{I,II}+X_{II,II}^TJX_{II,II}&=&J.
\end{eqnarray}

According to Section \ref{overzicht} we need to calculate an embedding for the generators $X_{ij}$ into the algebra of functions on the supergroup $OSp(m|2n)$. Equation \eqref{sheafsplit} implies that the supermanifold of $OSp(m|2n)$ is given by
\begin{eqnarray*}
\underline{OSp(m|2n)}=(O(m)\times Sp(2n),\cC^\infty_{O(m)\times Sp(2n)}\otimes \Lambda_{2mn}),
\end{eqnarray*}
since $\mbox{dim}_\mR(\osp_1)=2mn$. Therefore we must look for an embedding $X_{ij}\in\cC^\infty(O(m)\times Sp(2n))\otimes \Lambda_{2mn}$.

We introduce the $2mn$ independent Grassmann variables of  $\Lambda_{2mn}$ by labeling them as $\{\theta_{jk};j=1,\cdots,2n;k=1,\cdots,m\}$ which leads to the $2n\times m$ matrix $\theta$ and define $\hat{\theta}=\theta^TJ$. Define the $m\times m$ matrix $A$ and the $2n\times 2n$ matrix $B$ with entries in the even part of $\Lambda_{2mn}$, $\Lambda_{2mn}^{(ev)}$ as (finite) Taylor expansions,
\begin{eqnarray}
\label{defAB}
A=\sqrt{I_m-\hat{\theta}\theta}&\mbox{and}&B=\sqrt{I_{2n}-\theta\hat{\theta}}.
\end{eqnarray}
These matrices satisfy the following straightforward properties
\begin{eqnarray}
\label{eigB1}
B^TJ&=&JB\quad\mbox{and}\\
\label{eigB2}
\hat\theta B&=&A\hat\theta.
\end{eqnarray}

Using these definitions we can construct an embedding of the matrix elements, which will be needed to calculate equation \eqref{GenexprInt} for $OSp(m|2n)$.
\begin{theorem}
\label{Xifo}
Consider the matrices $x\in \left[\cC^\infty(O(m))\right]^{m\times m}$ and $y \in \left[\cC^\infty(Sp(2n))\right]^{2n\times 2n}$ of matrix elements of the fundamental representation of $O(m)$ and $Sp(2n)$. The matrix
\begin{eqnarray*}
X&=&\left( \begin{array}{cc} X_{I,I}&X_{I,II}\\  \vspace{-3.5mm} \\X_{II,I}&X_{II,II}
\end{array}
 \right)\in\left[\cC^\infty\left( O(m) \right)\otimes \cC^\infty(Sp(2n))\otimes \Lambda_{2mn}\right]^{(m+2n)\times (m+2n)}
\end{eqnarray*}
defined by
\begin{eqnarray*}
X_{I,I}=xA & &X_{I,II}=x\hat{\theta}y\\
X_{II,I}=\theta & &X_{II,II}=By
\end{eqnarray*}
satisfies the orthogonality relations of $OSp(m|2n)$ in equation \eqref{defrelations}.
\end{theorem}
\begin{proof}
First we remark that the matrix $\hat{\theta}\theta$ is symmetric, which implies that $A$ is symmetric. The first relation of \eqref{defrelations} follows from the following calculation:
\begin{eqnarray*}
(xA)^TxA+\theta^TJ\theta&=&A^2+\hat{\theta}\theta=I_m.
\end{eqnarray*}
To prove the second relation we use equation \eqref{eigB2}, which yields
\begin{eqnarray*}
-(xA)^Tx\hat{\theta}y+\theta^T JBy&=&-A\hat{\theta}y+A\hat{\theta}y=0.
\end{eqnarray*}
For the third relation we apply property \eqref{eigB1}, this leads to
\begin{eqnarray*}
-(x\hat{\theta}y)^Tx\hat{\theta}y+(By)^TJBy&=&y^TJ\theta\hat{\theta}y+y^TJB^2y=y^TJy=J,
\end{eqnarray*}
thus proving the theorem.
\end{proof}

Since the algebra generated by the $X_{ij}$ as defined in Theorem \ref{Xifo} includes the $x_{ij}$, $y_{kl}$ and $\theta_{ki}$ for $i,j=1,\cdots,m$, $k,l=1,\cdots,2n$, the comultiplication $\Delta$ on $\cA$ as in equation \eqref{defcomult} uniquely extends to a comultiplication on $\cC^\infty(O(m)\times Sp(2n))\otimes \Lambda_{2mn}$.

\begin{theorem}
\label{defsgr}
The Lie supergroup with supermanifold $\underline{OSp(m|2n)}$ equipped with the multiplication $\mu:\underline{OSp(m|2n)}\otimes \underline{OSp(m|2n)} \to \underline{OSp(m|2n)}$ and the involutive superdiffeomorphism $\nu:\underline{OSp(m|2n)} \to \underline{OSp(m|2n)}$ given below is the orthosymplectic supergroup $OSp(m|2n)$. The multiplication $\mu=(\mu_0,\mu^\sharp)$ is given by
\begin{eqnarray*}
\mu^\sharp(X_{ij})&=&\sum_{k=1}^{m+2n}(-1)^{([i]+[k])([k]+[j])}X_{ik}\otimes X_{kj},
\end{eqnarray*}
for $X_{ij}$ defined in Theorem \ref{Xifo}. The involutive superdiffeormorphism $\nu=(\nu_0,\nu^\sharp)$ is defined by
\begin{eqnarray*}
\nu^\sharp(X_{I,I})= X_{I,I}^T & & \nu^\sharp (X_{I,II})=-X_{II,I}^TJ\\
\nu^\sharp(X_{II,I})=-JX^T_{I,II} & &\nu^\sharp (X_{II,II})=-JX_{II,II}^TJ.
\end{eqnarray*}
\end{theorem}
\begin{proof}
This theorem follows essentially from the fact that the invariant derivations generate $\osp$ (see Theorem \ref{liealg1}) and the equivalence of categories between Harish-Chandra pairs and Lie supergroups.
\end{proof}

The corepresentation given in equation \eqref{corepOSp} now corresponds to a corepresentation of $OSp(m|2n)$ (more precisely to a corepresentation of the Hopf algebra of functions on $OSp(m|2n)$),
\begin{eqnarray*}
\chi:\,\mR^{m|2n}&\to&\cO_{OSp(m|2n)}(O(m)\times Sp(2n))\otimes \mR^{m|2n}.
\end{eqnarray*}
This justifies the terminology matrix elements for $X_{ij}$ defined in Theorem \ref{Xifo}.

The action of $\mu^\sharp$ on functions in $\Lambda_{2mn}$ or $\cC^\infty(O(m)\times Sp(2n))$ can be calculated from the following relations:
\begin{eqnarray}
\label{coacGrass}
\mu^\sharp (\theta_{ij})&=&\sum_{l=1}^m\theta_{il}\otimes(xA)_{lj}+\sum_{k=1}^{2n}(By)_{ik}\otimes \theta_{kj}\\
\label{coacx}
\mu^\sharp (x_{ij})&=&\sum_{l=1}^m\left(\sum_{t=1}^m(xA)_{it}\otimes (xA)_{tl}-\sum_{s=1}^{2n}(x\hat{\theta}y)_{is}\otimes \theta_{sl}\right)\mu^\sharp (A^{-1}_{lj})\\
\label{coacy}
\mu^\sharp (y_{ij})&=&\sum_{l=1}^m\mu^\sharp (B^{-1}_{il})\left(-\sum_{t=1}^m\theta_{lt}\otimes (x\hat{\theta}y)_{tj}+\sum_{s=1}^{2n}(By)_{ls}\otimes (By)_{sj}\right),
\end{eqnarray}
with $A$ and $B$ defined in equation \eqref{defAB}. Equation \eqref{coacGrass} gives the comultiplication on all the elements of the Grassmann algebra since $\mu^\sharp(\theta_F\theta_G)=\mu^\sharp(\theta_F)\mu^\sharp(\theta_G)$ for $\theta_F,\theta_G\in\Lambda_{2mn}$. In particular $(A^{-1})_{lj}=\left((I_m-\hat{\theta}\theta)^{-1/2} \right)_{lj}\in\Lambda_{2mn}$, appearing in equation \eqref{coacx} could in principle be calculated in that way. Equation \eqref{coacx} follows immediately from $\mu^\sharp(X_{il})=\sum_{j=1}^m\mu^\sharp(x_{ij})\mu^\sharp(A_{jl})$, so $\mu^\sharp(x_{ij})=\sum_{l=1}^m\mu^\sharp(X_{il})\mu^\sharp((A^{-1})_{lj})$.

\subsection{Action of the Harish-Chandra pair $(O(m)\times Sp(2n),\mathfrak{osp}(m|2n))$}
Now that the comultiplication on elements of $\cC^\infty(O(m)\times Sp(2n))\otimes \Lambda_{2mn}$ is known, we can use this to calculate the action of $O(m)\times Sp(2n)$ on these functions and to find an expression for the invariant derivations on $\cC^\infty(O(m)\times Sp(2n))\otimes \Lambda_{2mn}$. This corresponds to an action of the Harish-Chandra pair $(O(m)\times Sp(2n), \osp)$ on the supermanifold $(O(m)\times Sp(2n),\cC^\infty_{O(m)\times Sp(2n)}\otimes \Lambda_{2mn})$.

The multiplication $\mu$ leads to a left or right corepresentation of the Hopf algebra $\cC^\infty({O(m)\times Sp(2n)})$ by restricting $\delta^\sharp:\cO_{OSp(m|2n)}\to\cC^\infty_{O(m)\times Sp(2n)}$, see equation \eqref{projbos}.
\begin{lemma}
\label{Liegract}
The left co-action of $O(m)\times Sp(2n)$ on $\cO_{OSp(m|2n)}$, $\varphi^\sharp=(\delta^\sharp\otimes id^\sharp)\circ\mu^\sharp$, satisfies
\begin{eqnarray*}
\varphi^\sharp(\theta_{ij})=\sum_{k=1}^{2n} y_{ik}\otimes \theta_{kj},\quad\varphi^\sharp(x_{ij})=\sum_{k=1}^mx_{ik}\otimes x_{kj}\quad\mbox{and}\quad \varphi^\sharp(y_{ij})=\sum_{k=1}^{2n} y_{ik}\otimes y_{kj}.
\end{eqnarray*}
\end{lemma}
\begin{proof}
The first equation is a direct consequence of formula \eqref{coacGrass}. This first equation also shows immediately that $\hat{\theta}\theta$ is $O(m)\times Sp(2n)$-invariant:
\begin{eqnarray}
\label{invthetatheta}
(\delta^\sharp\otimes id^\sharp)\circ\mu^\sharp\left((\hat{\theta}\theta)_{ij}\right)&=&1\otimes (\hat{\theta}\theta)_{ij}.
\end{eqnarray}
This implies that $A^{-1}=(I_m-\hat{\theta}\theta)^{-\frac{1}{2}}$ is $O(m)\times Sp(2n)$-invariant. Equation \eqref{coacx} therefore yields
\begin{eqnarray*}
(\delta^\sharp\otimes id^\sharp)\circ\mu^\sharp(x_{ij})&=&\sum_{l=1}^m\sum_{k=1}^m(x_{ik}\otimes (xA)_{kl})(1\otimes A^{-1}_{lj})=\sum_{k=1}^mx_{ik}\otimes x_{kj}.
\end{eqnarray*}
The first equation in the lemma also yields
\begin{eqnarray*}
(\delta^\sharp\otimes id^\sharp)\circ\mu^\sharp\left((\theta\hat{\theta})_{ij}\right)&=&-y_{is}(Jy^TJ)_{tj}\otimes (\theta\hat{\theta})_{st}
\end{eqnarray*}
from which the formula
\begin{eqnarray*}
(\delta^\sharp\otimes id^\sharp)\circ\mu^\sharp(B^{-1})_{ij}&=&-y_{is}(Jy^TJ)_{tj}\otimes (B^{-1})_{st}
\end{eqnarray*}
follows. The last formula in the lemma is a direct consequence of this result and equation \eqref{coacy}.
\end{proof}

The algebra $\mathfrak{o}(m)$ is generated by the derivations $L_{ij}\in$ Der$\cC^\infty(O(m))$, $1\le i,j\le m$, defined by
\begin{eqnarray*}
L_{ij}(x_{kl})&=&\delta_{jk}x_{il}-\delta_{ik}x_{jl}.
\end{eqnarray*}
The algebra $\mathfrak{sp}(2n)$ is generated by the derivations $L_{i+m,j+m}\in$ Der$\cC^\infty(Sp(2n))$, $1\le i, j\le 2n$, defined by
\begin{eqnarray*}
L_{i+m,j+m}(y_{kl})&=&J_{kj}y_{il}+J_{ki}y_{jl}.
\end{eqnarray*}

Lemma \ref{defLiealg2} for this case means that the Lie superalgebra of invariant superderivations of $OSp(m|2n)$ is generated by the derivations $K_{ij}$, $1\le i\le j\le m+2n$ defined as
\begin{equation}
\label{defK}
K_{ij}=\begin{cases}
\left(\delta^\sharp_0\circ L_{ij}\otimes id^\sharp\right)\circ\mu^\sharp & 1\le i\le j\le m\\
\sum_{l=1}^{2n}J_{j-m,l}\left(\delta^\sharp_0\circ \partial_{\theta_{li}}\otimes id^\sharp\right)\circ\mu^\sharp & 1\le i\le m, \quad m+1\le j\le m+2n\\
\left(\delta^\sharp_0\circ L_{ij}\otimes id^\sharp\right)\circ\mu^\sharp & m+1\le i\le j\le m+2n.
\end{cases}
\end{equation}
Straightforward but tedious calculations (which have to be performed independently for the three different types of derivations and the four different types of matrix elements) show that these derivations satisfy
\begin{eqnarray}
\label{KonX}
K_{\alpha\beta}(X_{\gamma\delta})&=&(-1)^{(1+[\delta])([\alpha]+[\beta])}\left(g_{\gamma\beta}X_{\alpha\delta}-(-1)^{[\alpha][\beta]}g_{\gamma\alpha}X_{\beta\delta}\right).
\end{eqnarray}
As an example we calculate the case $\alpha,\beta,\gamma,\delta\le m$.
\begin{eqnarray*}
K_{ij}(X_{kl})&=&\sum_{s=1}^m\left(\delta^\sharp_0\circ L_{ij}(xA)_{ks}\right)X_{sl}-\sum_{t=1}^{2n}\left(\delta^\sharp_0\circ L_{ij}(x\hat{\theta}y)_{kt}\right)X_{t+m,l}\\
&=&\delta_{jk}X_{il}-\delta_{ik}X_{jl}.
\end{eqnarray*}
We also define the derivations $K_{\alpha\beta}$ for $\alpha> \beta$ by formula \eqref{KonX}. This immediately implies that $K_{\alpha\beta}=-(-1)^{[\alpha][\beta]}K_{\beta\alpha}$.

\begin{theorem}
\label{liealg1}
The super-derivations $K_{ij}\in Der\cO_{OSp(m|2n)}(O(m)\times Sp(2n))$ defined in equation \eqref{defK} or \eqref{KonX} satisfy the commutation relations of the standard generators in equation \eqref{stanosp} and therefore generate the Lie superalgebra $\mathfrak{osp}(m|2n)$.
\end{theorem}
\begin{proof}
derivations are completely determined by their action on $x$, $y$ and $\theta$, which is equivalent to knowing their action on $X_{\gamma\delta}$ for $1\le \gamma,\delta\le m+2n$. It is clear from equation \eqref{KonX} that the derivations do not mix up $X_{\gamma\delta}$ with $\delta\le m$ and $\delta > m$. The case $\delta> m$ is given by
\begin{eqnarray*}
K_{\alpha\beta}(X_{\gamma \delta})&=&\left(g_{\gamma\beta}X_{\alpha \delta}-(-1)^{[\alpha][\beta]}g_{\gamma\alpha}X_{\beta \delta}\right).
\end{eqnarray*}
Comparison with equation \eqref{defgenosp} shows that the supercommutator evaluated on the $X_{\gamma\delta}$ for $\delta> m$ gives the correct formula. The case $\delta\le m$ is the same except for the fact that the odd generators of the algebra obtain a minus sign, which is unimportant due to the gradation preserving property of the Lie superbracket.
\end{proof}

The $\cO_{OSp(m|2n)}(O(m)\times Sp(2n))$-module
\begin{eqnarray*}
Der\cO_{OSp(m|2n)}(O(m)\times Sp(2n))&=& Der\left( \cC^\infty(O(m)\times Sp(2n))\otimes \Lambda_{2mn}\right)
\end{eqnarray*}
is generated by $L_{ij}$ for $i,j=1,\cdots,m$ and $i,j=m+1,\cdots,m+2n$ and $\partial_{\theta_{ij}}$ for $i=1,\cdots,2n$ and $j=1,\cdots,2n$.  The invariant derivations on $OSp(m|2n)$ defined in equation \eqref{defK} can therefore be expressed in terms of the invariant derivations on $O(m)\times Sp(2n)$ and the Grassmann derivations. This leads to the following theorem.
\begin{theorem}
\label{exposp}
The even elements of $Der\cO_{OSp(m|2n)}(O(m)\times Sp(2n))$ defined in equation \eqref{defK} satisfy the following relations:
\begin{eqnarray*}
K_{ij}&=&L_{ij}\qquad\mbox{for}\quad 1\le i <j\le m \\
K_{i+m,j+m}&=&L_{i+m,j+m}+\sum_{l=1}^m\sum_{p=1}^{2n}\left(\theta_{il}J_{jp}+\theta_{jl}J_{ip}\right)\partial_{\theta_{pl}}\qquad \mbox{for} \quad 1 \le i\le j\le 2n.
\end{eqnarray*}
The odd elements satisfy for $1\le i\le m$ and $1\le j\le 2n$,
\begin{eqnarray*}
K_{i,j+m}&=&\sum_{t=1}^m\sum_{p=1}^{2n}(xA)_{it}J_{jp}\partial_{\theta_{pt}}+\frac{1}{2}\sum_{t=1}^m(xA^{-1}\theta^T)_{tj}L_{ti}\\
&+&\frac{1}{2}\sum_{s,t=1}^{2n}\left((JB^{-1}J)_{tj}(x\hat{\theta})_{is}-\sum_{u,p=1}^{2n}\sum_{r=1}^m(xA)_{ir} (JB^{-1})_{tu}J_{jp}\left(\partial_{\theta_{pr}} B_{us}\right)\right)L_{s+m,t+m}\\
&-&\frac{1}{2}\sum_{s,t,r,u=1}^m\sum_{p=1}^{2n}(xA)_{it}(xA^{-1})_{us}J_{jp}\partial_{\theta_{pt}}\left((xA)_{rs}\right)L_{ur}
\end{eqnarray*}
with $A$ and $B$ as defined in equation \eqref{defAB}.
\end{theorem}
\begin{proof}
The expression for the even generators follows from a straightforward calculation. It also corresponds to the derivative of the $O(m)\times Sp(2n)$-action in Lemma \ref{Liegract}.

In order to prove that the expression for the odd generators holds, it is sufficient to show that left-hand and right-hand side have the same value when evaluated on the functions $x$, $y$ and $\theta$. This is equivalent to having the same value on the functions $xA$, $By$ and $\theta$ since $A$ and $B$ (given in equation \eqref{defAB}) are invertible. The results for the left-hand sides are given in equation \eqref{KonX}. The theorem is therefore proven if the proposed expression for $K_{i,j+m}$ on the right-hand side gives the same results.

We evaluate the proposed expression for $K_{ij}$ on these functions. The case $\theta$ is straightforward,
\begin{eqnarray*}
K_{i,j+m}(\theta_{kl})&=&\sum_{t=1}^m\sum_{p=1}^{2n}(xA)_{it}J_{jp}\partial_{\theta_{pt}}(\theta_{kl})=J_{jk}(xA)_{il}=-g_{k+m,j+m}X_{il}.
\end{eqnarray*}
The case $xA$ is calculated as
\begin{eqnarray*}
& &K_{ij+m}((xA)_{kl})\\
&=&\sum_{t=1}^m\sum_{p=1}^{2n}(xA)_{it}J_{jp}\partial_{\theta_{pt}}((xA)_{kl})+\frac{1}{2}\theta_{jl}\delta_{ik}-\frac{1}{2}(xA^{-1}\theta^T)_{kj}(xA)_{il}\\
&+&\frac{1}{2}\sum_{s,t,r=1}^m\sum_{p=1}^{2n}(xA)_{it}(xA^{-1})_{ks}J_{jp}\partial_{\theta_{pt}}\left(A_{rs}\right)A_{rl}-\frac{1}{2}\sum_{t=1}^m\sum_{p=1}^{2n}(xA)_{it}J_{jp}\partial_{\theta_{pt}}\left((xA)_{kl}\right).
\end{eqnarray*}
Subtract the last term from the first and replace $(xA)_{kl}$ by $\sum_{r,s=1}^m(A_{rl})(A_{rs}(xA^{-1})_{ks}))$. We obtain 
\begin{eqnarray*}
&&\frac{1}{2}\sum_{t=1}^m(xA)_{it}\sum_{r,s=1}^m\sum_{p=1}^{2n}J_{jp}\partial_{\theta_{pt}}(A_{rl})(A_{rs}(xA^{-1})_{ks}))+\frac{1}{2}\theta_{jl}\delta_{ik}-\frac{1}{2}(xA^{-1}\theta^T)_{kj}(xA)_{il}\\
&+&\frac{1}{2}\sum_{s,t,r=1}^m\sum_{p=1}^{2n}(xA)_{it}(xA^{-1})_{ks}J_{jp}\partial_{\theta_{pt}}\left(A_{rs}\right)A_{rl}\\
&=&\frac{1}{2}\sum_{t,s=1}^m\sum_{p=1}^{2n}(xA)_{it}J_{jp}\partial_{\theta_{pt}}(A^2_{ls})(xA^{-1})_{ks}+\frac{1}{2}\theta_{jl}\delta_{ik}-\frac{1}{2}(xA^{-1}\theta^T)_{kj}(xA)_{il}\\
&=&\delta_{ik}\theta_{jl}=g_{ki}X_{j+m,l}.
\end{eqnarray*}


The case $By$ follows from a similar but more complicated calculation,
\begin{eqnarray*}
& &K_{ij+m}((By)_{kl})\\
& =&\sum_{t=1}^m\sum_{p=1}^{2n}(xA)_{it}J_{jp}\partial_{\theta_{pt}}((By)_{kl})-\frac{1}{2}\sum_{s,t=1}^{2n}(JB^{-1}J)_{tj}(x\hat{\theta})_{is}\left((BJ)_{kt}y_{sl}+(BJ)_{ks}y_{tl}\right)\\
&+&\frac{1}{2}\sum_{s,t=1}^{2n}\sum_{u,p=1}^{2n}\sum_{r=1}^m(xA)_{ir} (JB^{-1})_{tu}J_{jp}\left(\partial_{\theta_{pr}} B_{us}\right)\left((BJ)_{kt}y_{sl}+(BJ)_{ks}y_{tl}\right)\\
&=&\sum_{t=1}^m\sum_{p=1}^{2n}(xA)_{it}J_{jp}\partial_{\theta_{pt}}((By)_{kl})+\frac{1}{2}\sum_{s,t=1}^{2n}(JB^{-1}J)_{tj}(x\hat{\theta})_{is}(BJ)_{ks}y_{tl}\\
&+&\frac{1}{2}J_{kj}(x\hat{\theta}y)_{il}-\frac{1}{2}\sum_{p=1}^{2n}\sum_{r=1}^m(xA)_{ir} J_{jp}\left(\partial_{\theta_{pr}} (By)_{kl}\right)\\
&+&\frac{1}{2}\sum_{s,t=1}^{2n}\sum_{u,p=1}^{2n}\sum_{r=1}^m(xA)_{ir} (JB^{-1})_{tu}J_{jp}\left(\partial_{\theta_{pr}} B_{us}\right)(BJ)_{ks}y_{tl}.
\end{eqnarray*}
Subtract the fourth term from the first and replace $(By)_{kl}$ by $\sum_{s,r,u=1}^{2n}(BJ)_{ks} (JB^{-1})_{ru}B_{us}y_{rl}$. We obatin
\begin{eqnarray*}
&&\frac{1}{2}\sum_{t=1}^m\sum_{p,s,r,u=1}^{2n}(xA)_{it}J_{jp}\partial_{\theta_{pt}}((BJ)_{ks})(JB^{-1})_{ru}B_{us}y_{rl}+\frac{1}{2}\sum_{s,t=1}^{2n}(JB^{-1}J)_{tj}(x\hat{\theta})_{is}(BJ)_{ks}y_{tl}\\
&+&\frac{1}{2}J_{kj}(x\hat{\theta}y)_{il}+\frac{1}{2}\sum_{s,t=1}^{2n}\sum_{u,p=1}^{2n}\sum_{r=1}^m(xA)_{ir} (JB^{-1})_{tu}J_{jp}\left(\partial_{\theta_{pr}} B_{us}\right)(BJ)_{ks}y_{tl}\\
&=&\frac{1}{2}\sum_{t=1}^m\sum_{p,r,u=1}^{2n}(xA)_{it}J_{jp}\partial_{\theta_{pt}}((BJB^T)_{ku})(JB^{-1})_{ru}y_{rl}+\frac{1}{2}J_{kj}(x\hat{\theta}y)_{il}\\
&+&\frac{1}{2}\sum_{s,t=1}^{2n}(JB^{-1}J)_{tj}(x\hat{\theta})_{is}(BJ)_{ks}y_{tl}.
\end{eqnarray*}
By using equations \eqref{eigB1} and \eqref{eigB2}, we can further simplify the obtained expression to
\begin{eqnarray*}
&&-\frac{1}{2}(xA\hat{\theta}B^{-1}y)_{il}J_{jk}+\frac{1}{2}\sum_{r=1}^{2n}(xA\theta^T)_{ik}(JB^{-1}J)_{rj}y_{rl}+\frac{1}{2}J_{kj}(x\hat{\theta}y)_{il}\\
&&+\frac{1}{2}\sum_{t=1}^{2n}(JB^{-1}J)_{tj}(x\theta^TB^T)_{ik}y_{tl}=-J_{jk}(x\hat{\theta}y)_{il}=g_{k+m,j+m}X_{i,l+m}.
\end{eqnarray*}
This completes the proof of the theorem.
\end{proof}

\begin{remark}
A choice of embedding, different from Theorem \ref{Xifo}, is given by
\begin{eqnarray*}
X&=&\left( \begin{array}{cc} xA&x\hat{\theta}\\  \vspace{-3.5mm} \\y\theta&yB\end{array}\right)
\end{eqnarray*}
where $A$ and $B$ are still given by equation \eqref{defAB}. This choice would have the advantage that both $K_{ij}=L_{ij}$ and $K_{i+m,j+m}=L_{i+m,j+m}$ would hold in Theorem \ref{exposp}. The expression for $K_{i,j+m}$ would however be harder to obtain. Another choice could be
\begin{eqnarray*}
X&=&\left( \begin{array}{cc} Ax&\hat{\theta}\\  \vspace{-3.5mm} \\y\theta x&yB\end{array}\right),
\end{eqnarray*}
which would yield results very similar to Theorem \ref{exposp}.
\end{remark}

\subsection{Invariant integration on $OSp(m|2n)$}

The action of $(O(m)\times Sp(2n),\osp)$ on $\underline{OSp(m|2n)}$ as calculated in the previous section can now be used to construct the invariant integral on $OSp(m|2n)$. This integral can also be constructed iteratively from the supersphere integral (see e.g. \cite{MR2539324, DBE1}) according to the formula before Corollary $3$ in \cite{CouOSp}. That approach could also be useful to calculate integrals of monomials by extending the approach to $O(m)$ in \cite{MR2385262}.

Before we derive the explicit expression for the invariant integral on $OSp(m|2n)$ we need the following technical lemma.
\begin{lemma}
\label{diffeqdet}
The differential equation
\begin{eqnarray*}
\partial_{\theta_{ij}}\det(I-\hat{\theta}\theta)&=&2\left((I-\hat{\theta}\theta)^{-1}\hat{\theta}\right)_{ji}\det(I-\hat{\theta}\theta)
\end{eqnarray*}
holds for all $i, j$ satisfying $1\le i\le 2n$ and $1\le j\le m$.
\end{lemma}

\begin{proof}
We use the well-known formula $\det N=\exp(-\mbox{tr} M)$ with $M=-\ln(N)$. This implies
\begin{eqnarray*}
\det(I-\hat{\theta}\theta)&=&\exp\left(-\sum_{k=1}^\infty \frac{\mbox{tr}(\hat{\theta}\theta)^{k}}{k}\right)
\end{eqnarray*}
where the summations are actually finite because of the nilpotency of the Grassmann variables. A straightforward calculation then shows
\begin{eqnarray*}
\partial_{\theta_{ij}}\det(I-\hat{\theta}\theta)&=&\exp\left(-\sum_{k=1}^\infty \frac{\mbox{tr}(\hat{\theta}\theta)^{k}}{k}\right)\left(2\sum_{k=1}^\infty \sum_{l=1}^m (\hat{\theta}\theta)^{k-1}_{jl}\hat{\theta}_{li}\right),
\end{eqnarray*}
which proves the lemma.
\end{proof}
Lemma \ref{diffeqdet} implies that for $A$ defined in equation \eqref{defAB},
\begin{eqnarray}
\nonumber
\sum_{t=1}^mA^2_{lt}\partial_{\theta_{jt}}(\det A)^{-1}&=&\sum_{t=1}^mA^2_{lt}\partial_{\theta_{jt}}(\det I-\hat{\theta}\theta)^{-\frac{1}{2}}\\
\label{diffdetinv}
&=&-\frac{1}{2}\sum_{t=1}^mA^2_{lt}(\det I-\hat{\theta}\theta)^{-\frac{3}{2}}2\left((I-\hat{\theta}\theta)^{-1}\hat{\theta}\right)_{tj}\det(I-\hat{\theta}\theta)\\
\nonumber
&=&-\hat{\theta}_{lj}(\det A)^{-1}.
\end{eqnarray}

The obtained results can now be used to prove the main result of integration on $OSp(m|2n)$. We introduce the ordinary Berezin integral on $\Lambda_{2mn}$, see e.g. \cite{Berezin, MR778559, MR574696} as
\begin{eqnarray*}
\int_{B_{2mn}}&=&\partial_{\theta_{2n,m}}\partial_{\theta_{2n,m-1}}\cdots\partial_{\theta_{2n,1}}\partial_{\theta_{2n-1,m}}\cdots\partial_{\theta_{1,1}}.
\end{eqnarray*}

Since $Sp(2n)$ is not compact we need to make a distinction between two types of functions on $OSp(m|2n)$. We denote the functions with compact support as
\[\cO_{OSp(m|2n)}(O(m)\times Sp(2n))_c=\cC^\infty_c(O(m)\times Sp(2n))\otimes\Lambda_{2mn}\]
and the Hopf algebra of regular functions as
\[
\begin{aligned}
\cO_{OSp(m|2n)}(O(m)\times Sp(2n))_0&=\cC^\infty(O(m)\times Sp(2n))_0\otimes\Lambda_{2mn}\\
&=Alg(x_{ij},y_{kl},\theta_{ki})=Alg(X_{ij}).
\end{aligned}
\]

There is a unique Haar measure on $Sp(2n)$,
\begin{eqnarray*}
 \int_{Sp(2n)} &:& \cC^\infty_c(Sp(2n))\to\mR.
\end{eqnarray*}
There is also a unique invariant Hopf-algebraical integral on $Sp(2n)$, which leads to the invariant linear functional
\begin{eqnarray*}
 \int_{Sp(2n)_0} &:& \cC^\infty(Sp(2n))_0\to\mR.
\end{eqnarray*}
This integral can be identified with the Haar measure on $USp(2n)$,
\begin{eqnarray*}
\int_{Sp(2n)_0}f(y_{kl})&=&\int_{USp(2n)}f(z_{kl}),
\end{eqnarray*}
where we make a substitution to the monomials on $USp(2n)$ determined by $z^TJz=J$ and $z^\dagger z=I_{2n}$.

\begin{theorem}\label{InvInt}
The unique invariant integral on $OSp(m|2n)$,
\begin{eqnarray*}
\int_{OSp(m|2n)_0}:\cO_{OSp(m|2n)}(O(m)\times Sp(2n))_0=\cC^\infty(O(m)\times Sp(2n))_0\otimes \Lambda_{2mn}\to\mR
\end{eqnarray*}
is given by
\begin{eqnarray*}
\int_{OSp(m|2n)_0}\cdot&=&\int_{O(m)\times Sp(2n)_0}\int_{B_{2mn}}\left(\det A\right)^{-1}\cdot,
\end{eqnarray*}
where $A$ is defined in equation \eqref{defAB}.
\end{theorem}

\begin{remark}
\label{remarkintint}
Before we prove this formula we show how this integration should be interpreted. Most relevant integrands will be functions expressed in terms of the matrix elements $X_{ij}$ of the fundamental representation of $OSp(m|2n)$ on $\mR^{m|2n}$, see formula \eqref{corepOSp}. The formula in Theorem \ref{InvInt} then states that for a function $f(X)=f(X_{11},\cdots,X_{m+2n,m+2n})$
\begin{eqnarray*}
\int_{OSp(m|2n)_0}f&=&\int_{O(m)\times Sp(2n)_0,x\times y}\int_{B_{2mn},\theta}\left(\det A(\theta)\right)^{-1}f\left(X(x,y,\theta)\right),
\end{eqnarray*}
with $X(x,y,\theta)$ given in Theorem \ref{Xifo}.
\end{remark}

\begin{proof}[Proof of Theorem \ref{InvInt}]
The result in Theorem 1 in \cite{MR1845224} states that there is at most one linear functional $\int$ which satisfies $(id^\sharp\otimes\int)\circ\mu^\sharp=1_{OSp(m|2n)}\int$, so only the invariance of $\int_{OSp(m|2n)_0}$ needs to be proven.
By Lemma \ref{invHopf}, it suffices to show that the proposed expression is
\begin{itemize}
\item $O(m)\times Sp(2n)$-invariant for the action given in Lemma \ref{Liegract} and
\item satisfies $\int_{OSp(m|2n)_0}\circ K_{ij}=0$ for all $1\le i\le j\le m+2n$ with $K_{ij}$ given in Theorem \ref{exposp}.
\end{itemize}

The integral $\int_{O(m)\times Sp(2n)_0}$ is $O(m)\times Sp(2n)$-invariant and so is $\int_{B_{2mn}}$. The last statement can be seen immediately from the evaluation of $\int_{B_{2mn}}$ on monomials in $\Lambda_{2mn}$, which only gives a non-zero result on the unique highest order monomial. The $O(m)\times Sp(2n)$-invariance of $\int_{OSp(m|2n)_0}$ then follows from the invariance of $\hat{\theta}\theta$ in equation \eqref{invthetatheta},
\begin{eqnarray*}
\left(\delta^\sharp\otimes \int_{OSp(m|2n)_0}\right)\circ\mu^\sharp(f)&=&\left(\delta^\sharp\otimes \int_{O(m)\times Sp(2n)_0}\int_{B_{2mn}}\right)\circ\mu^\sharp(\det(A^{-1})f)\\
&=&\int_{O(m)\times Sp(2n)_0}\int_{B_{2mn}}\det(A^{-1})f= \int_{OSp(m|2n)_0}f.
\end{eqnarray*}
This also implies that the relation $\int_{OSp(m|2n)_0}\circ K_{ij}=0$ holds for the even generators.

The condition $\int_{OSp(m|2n)_0}\circ K_{ij}=0$ for the odd elements of $\mathfrak{osp}(m|2n)$ (given in Theorem \ref{exposp}) yields a differential equation for $(\det A)^{-1}$ by partial integration,
\begin{eqnarray*}
& &\sum_{t=1}^m\sum_{p=1}^{2n}\partial_{\theta_{pt}}(xA)_{it}J_{jp}\det A^{-1}+\frac{1}{2}\sum_{t=1}^mL_{ti}(xA^{-1}\theta^T)_{tj}\det A^{-1}\\
&=&\frac{1}{2}\left(\sum_{s,t,r,u=1}^m\sum_{p=1}^{2n}L_{ur}(xA)_{it}(xA^{-1})_{us}J_{jp}\partial_{\theta_{pt}}\left((xA)_{rs}\right)\right) \det A^{-1}.
\end{eqnarray*}

The invariance of the proposed expression for the integral is therefore satisfied if this differential equation holds. Straightforward calculations show that this differential equation can be simplified to
\begin{eqnarray*}
& &\left(\sum_{t=1}^m\sum_{p=1}^{2n}\frac{1}{2}\left(\partial_{\theta_{pt}}(xA)_{it}\right)J_{jp}\right) \det A^{-1}+\sum_{t=1}^m\sum_{p=1}^{2n}(xA)_{it}J_{jp}\partial_{\theta_{pt}}\det A^{-1}\\
&=&\frac{1}{2}\left((m-1)\left(xA^{-1}\theta^T\right)_{ij}-\sum_{s,t,r=1}^m\sum_{p=1}^{2n}(xA)_{rt}(xA^{-1})_{is}J_{jp}\left(\partial_{\theta_{pt}}(xA)_{rs}\right)\right) \det A^{-1}.
\end{eqnarray*}

We obtain a $m\times 2n$ matrix of differential equations. This set of equations is equivalent to the set obtained by multiplying the equations with the invertible matrices $(xA)_{il}$ and $J_{jk}$ and summing over $i$ and $j$,
\begin{eqnarray*}
& &\left(\frac{1}{2}\sum_{t,i=1}^mA_{il}\left(\partial_{\theta_{kt}}A_{it}\right)+\sum_{t=1}^mA^2_{lt}\partial_{\theta_{kt}}-\frac{m-1}{2}\hat{\theta}_{lk}+\frac{1}{2}\sum_{t,r=1}^mA_{rt}\left(\partial_{\theta_{kt}}A_{rl}\right)\right)\det A^{-1}\\
&=&\left(\sum_{t=1}^mA^2_{lt}\partial_{\theta_{kt}}-\frac{m-1}{2}\hat{\theta}_{lk}+\frac{1}{2}\sum_{t=1}^m\left(\partial_{\theta_{kt}}A^2_{lt}\right)\right)\det A^{-1}\\
&=&\left(\sum_{t=1}^mA^2_{lt}\partial_{\theta_{kt}}+\hat{\theta}_{lk}\right)\det A^{-1}=0
\end{eqnarray*}
for every $1\le k \le 2n$ and $1\le l\le m$. This is exactly differential equation \eqref{diffdetinv} that $\det A^{-1}$ satisfies.
\end{proof}

\begin{remark}
We have expressed the integral in terms of the matrix elements of the fundamental representation for $OSp(m|2n)$. Another choice for the fundamental representation would be where $Sp(2n)$ acts on the even part and $O(m)$ on the odd part. This is the natural $OSp(m|2n)$-action on $\mR^{2n|m}$. The matrix elements are denoted by $Y_{ij}$ and can again be written in block form as $Y=\left( \begin{array}{cc} Y_{I,I}&Y_{I,II}\\  \vspace{-3.5mm} \\Y_{II,I}&Y_{II,II}\end{array} \right)$, where now $Y_{II}$ has dimension $2n\times 2n$ and so on. The following relations are then equivalent to \eqref{defrelations}:
\begin{eqnarray*}
Y_{I,I}^TJY_{I,I}+Y_{II,I}^TY_{II,I}&=&J,\\
-Y_{I,I}^TJY_{I,II}+Y_{II,I}^TY_{II,II}&=&0,\\
-Y_{I,II}^TJY_{I,II}+Y_{II,II}^TY_{II,II}&=&I_m.
\end{eqnarray*}
We can construct a new embedding
\begin{eqnarray*}
\left( \begin{array}{cc} Y_{I,I}&Y_{I,II}\\  \vspace{-3.5mm} \\Y_{II,I}&Y_{II,II}\end{array} \right)&=&\left( \begin{array}{cc} \sqrt{I_{2n}+\theta\hat{\theta}}\,y& \theta\\  \vspace{-3.5mm} \\-x\hat{\theta}y&x\sqrt{I_m+\hat{\theta}\theta}\end{array} \right).
\end{eqnarray*}
Since the calculations are so similar we can immediately give the result. For a function $f$ of the matrix elements $Y_{ij}$, the invariant integral is given by
\begin{eqnarray*}
\int_{OSp(m|2n)_0}f&=&\int_{O(m)\times Sp(2n)_0,x\times y}\int_{B_{2mn},\theta}\left(\det (I_m+\hat{\theta}\theta)\right)^{-1/2}f\left(Y(x,y,\theta)\right).
\end{eqnarray*}
\end{remark}

The expression obtained in Theorem \ref{InvInt} immediately yields the following result.
\begin{corollary}\label{Int-nondegeneracy}
The invariant integral on $OSp(m|2n)$ is non-degenerate in the sense that the supersymmetric bilinear form
\begin{eqnarray*}
\left(f,g\right)&=&\int_{OSp(m|2n)_0}fg
\end{eqnarray*}
on $\cO_{OSp(m|2n)}(O(m)\times Sp(2n))_0$ is non-degenerate.
\end{corollary}
\begin{proof}
Assume there is a function $f\in\cC^\infty(O(m)\times Sp(2n))_0\otimes \Lambda_{2mn}$ such that $\int_{OSp(m|2n)_0}fg=0$ for every $g\in\cC^\infty(O(m)\times Sp(2n))_0\otimes \Lambda_{2mn}$. The function $f$ is of the form
\begin{eqnarray*}
f&=&\sum_A f_A \theta_A
\end{eqnarray*}
with $\theta_A$ a basis of monomials for $\Lambda_{2mn}$ and $f_A\in\cC^\infty(O(m)\times Sp(2n))_0$. Choose one $A$ such that $f_A$ is nonzero and $\theta_A$ is of the lowest degree. There is a monomial $\theta_B$ in $\Lambda_{2mn} $ such that $\int_{B_{2mn}}\theta_A\theta_B$ is not zero. By taking $g=\det(A) h \theta_B$ with $h\in \cC^\infty(O(m)\times Sp(2n))_0$ we obtain
\begin{eqnarray*}
\int_{O(m)\times Sp(2n)_0}f_A h&=&0
\end{eqnarray*}
for every $h\in \cC^\infty(O(m)\times Sp(2n))_0$. This is a contradiction since $\int_{O(m)\times Sp(2n)_0}$ is non-degenerate on $\cC^\infty(O(m)\times Sp(2n))_0$ and $f_A$ was assumed to be non-zero.
\end{proof}

Similar arguments yield the following property of $\det A$.
\begin{corollary}
\label{detunique}
The differential equations \eqref{diffdetinv} for $f\in\Lambda_{2mn}$
\begin{eqnarray}
\label{diffdetinv2}
\sum_{t=1}^mA^2_{lt}\partial_{\theta_{jt}}f&=&-\hat{\theta}_{lj}f \qquad\mbox{ for all } 1\le l\le m,\quad 1\le j\le 2n
\end{eqnarray}
uniquely determine (up to a multiplicative constant) $f$ to be $\det A^{-1}$.
\end{corollary}

\begin{example}
In case $m=1$, we use the notation $\theta_j=\theta_{j1}$ and $\theta^2=\sum_{j,k=1}^{2n}\theta_j J_{jk}\theta_k$, the invariant integral on $OSp(1|2n)$ is then given by
\begin{eqnarray*}
\int_{OSp(1|2n)_0}&=&\int_{Sp(2n)_0}\int_{B_{2n}}(1-\theta^2)^{-\frac{1}{2}}.
\end{eqnarray*}
\end{example}
In this case we obtain
\[\int_{OSp(1|2n)_0}1=\int_{B_{2n}}\frac{\theta^{2n}}{n!}\frac{\Gamma\left(\frac{1}{2}+n\right)}{\Gamma\left(\frac{1}{2}\right)}=\frac{\Gamma\left(\frac{1}{2}+n\right)}{\Gamma\left(\frac{1}{2}\right)}\not=0.\]
In Proposition $2$ in \cite{MR1845224} it was proven that a Hopf superalgebra $\cA$ admits a left-invariant integral $\int$ with $\int 1_\cA\not=0$ if and only if all right $\cA$-comodules are completely reducible. This agrees with the fact that all $OSp(1|2n)$-representations are completely reducible, see \cite{MR0387363}.

Another case that can be easily simplified is $OSp(m|2)$.
\begin{theorem}
If $n=1$ the relation
\begin{eqnarray*}
\det(I-\hat{\theta}\theta)^{-\frac{1}{2}}&=&1+\frac{1}{2}\mbox{tr}\hat{\theta}\theta=1+\sum_{j=1}^m\theta_{1j}\theta_{2j}
\end{eqnarray*}
holds.
\end{theorem}
\begin{proof}
In order to prove this equality we show that the left-hand side above satisfies the same differential equation \eqref{diffdetinv2} that uniquely characterizes $(\det A)^{-1}$ according to Corollary \ref{detunique}:
\begin{eqnarray*}
\sum_{j=1}^mA^2_{kj}\partial_{\theta_{ij}}(1+\frac{1}{2}\mbox{tr}\hat{\theta}\theta)&=&-\hat{\theta}_{ki}+(\hat{\theta}\theta\hat{\theta})_{ki}\\
&=&-\hat{\theta}_{ki}+\sum_{j=1}^m\theta_{1k}\theta_{2j}\theta_{1j}J_{1i}-\sum_{j=1}^m\theta_{2k}\theta_{1j}\theta_{2j}J_{2i}\\
&=&-\hat{\theta}_{ki}\left(1+\frac{1}{2}\mbox{tr}\hat{\theta}\theta\right),
\end{eqnarray*}
which proves the lemma.
\end{proof}

\begin{example}
In case $n=1$ the invariant integral on $OSp(m|2)$ is given by
\begin{eqnarray*}
\int_{OSp(m|2)_0}\cdot&=&\int_{O(m)\times Sp(2)_0}\int_{B_{2m}}\left(1+\sum_{j=1}^m\theta_{1j}\theta_{2j}\right)\cdot.
\end{eqnarray*}
\end{example}
This implies $\int_{OSp(m|2)_0}1=0$ if $m>1$, which is required since, in that case, not all the $OSp(m|2)$-representations are completely reducible, see \cite{MR051963}.

\begin{remark}
The same calculations as in Theorem \ref{InvInt} lead to the conclusion that the invariant integration on $OSp(m|2n)$ for functions with compact support
\begin{eqnarray*}
\int_{OSp(m|2n)}:\cO_{OSp(m|2n)}(O(m)\times Sp(2n))_c=\cC^\infty_c(O(m)\times Sp(2n))\otimes \Lambda_{2mn}\to\mR
\end{eqnarray*}
is given by
\begin{eqnarray*}
\int_{OSp(m|2n)}\cdot&=&\int_{O(m)\times Sp(2n)}\int_{B_{2mn}}\left(\det A\right)^{-1}\cdot.
\end{eqnarray*}

\end{remark}

\section{The Lie supergroup $U(p|q)$}

\subsection{Definition}

The Lie supergroup $U(p|q)$ corresponds to the pair $(U(p)\times U(q),\mathfrak{u}(p|q))$, which is the real Harish-Chandra pair
\[(\cG_0,\mathfrak{g})\subset(GL(p;\mC)\times GL(q;\mC),\mathfrak{gl}(p|q;\mC))=(\mbox{Aut}(\mC^{p})\times\mbox{Aut}(\mC^q),\,\mbox{End}(\mC^{p|q}))\]
acting on $\mC^{p|q}$ such that
\begin{eqnarray*}
L\circ\left(\overline{\pi_0}\otimes \pi_0\right)(S) =L\quad\mbox{on}\quad \mC^{p|q}\otimes\mC^{p|q} & & \forall \,S\in \cG_0,\\
L\circ(\overline{\rho^\pi}\otimes \rho^\pi)(X)=0\quad\mbox{on}\quad \mC^{p|q}\otimes\mC^{p|q} & &\forall\, X\in \mathfrak{g},
\end{eqnarray*}
with $L$ defined in Lemma \ref{metru}.

The Lie group $U(p)$ is a real group and the smooth functions $\cC^\infty(U(p))$ are real functions $U(p)\to\mR$. The matrix elements $x_{ij}$ of the fundamental representation of $U(p)$ are however complex-valued functions, $x_{ij}\in \cC^\infty(U(p))\oplus \,i\,\cC^\infty(U(p))$. The algebra generated by $x_{ij}$ and $\overline{x}_{ij}$ (subject to the relations $\sum_{i}\overline{x}_{ij}x_{ik}=\delta_{jk}$) can still be seen as a real algebra, although it is embedded in $\mC\otimes\cC^\infty(U(p))$. Nevertheless, the invariant integral on $U(p)$,
\begin{eqnarray*}
\int_{U(p)}&:&\cC^\infty(U(p))\to\mR
\end{eqnarray*}
still satisfies the property that $\int_{U(p)}$ evaluated on elements of the algebra generated by $x_{ij}$ and $\overline{x}_{ij}$ gives real values.

Similar to the case $OSp(m|2n)$, we can again construct the appropriate real super bialgebra $\cA$ for $U(p|q)$. However, as for the Hopf algebra of functions on $U(p)$, the algebra generated by the matrix elements for $U(p|q)$ corresponds to the product of the algebra of functions and the complex numbers. The algebra $\mC\otimes \cA$ is therefore generated by $X_{jk}=X_{jk}^1+iX_{jk}^2$ and their complex conjugates $\overline{X_{jk}}=X_{jk}^1-iX_{jk}^2$ subject to some relations. So $\cA$ is generated by $\{X_{jk}^1,X_{jk}^2\}$ subject to those relations. The comultiplication on $\cA$ is defined by equation \eqref{defcomult} and the condition $\Delta(\overline{f})=\overline{\Delta(f)}$ for $f\in\mC\otimes \cA$.

Again we write $X$ as a block matrix
 $X=\left( \begin{array}{cc} X_{I,I}&X_{I,II}\\  \vspace{-3.5mm} \\X_{II,I}&X_{II,II} \end{array}\right)$
according to the gradation. Similarly to Section \ref{Definition1}, the natural representation of the Harish-Chandra pair $(U(p)\times U(q),\mathfrak{u}(p|q))$ leads to 3 independent relations for the submatrices:
\begin{eqnarray}
\nonumber
X_{I,I}^\dagger X_{I,I}+iX_{II,I}^\dagger X_{II,I}&=&I_p,\\
\label{defrelations2}
-X_{I,I}^\dagger X_{I,II}+iX_{II,I}^\dagger X_{II,II}&=&0,\\
\nonumber
-X_{I,II}^\dagger X_{I,II}+iX_{II,II}^\dagger X_{II,II}&=&i I_q.
\end{eqnarray}

According to Section \ref{overzicht}, we need to calculate an embedding for the generators $X_{ij}$ into the algebra of functions on the supergroup $U(p|q)$. Equation \eqref{sheafsplit} implies that the supermanifold of $U(p|q)$ is given by
\begin{eqnarray*}
\underline{U(p|q)}=(U(p)\times U(q),\cC^\infty_{U(p)\times U(q)}\otimes \Lambda_{2pq})
\end{eqnarray*}
since $\mbox{dim}_\mR(\mathfrak{u}(p|q)_1)=2pq$. Therefore we must look for an embedding
\begin{eqnarray*}
X_{ij}&\in&\mC\otimes\cC^\infty(U(p)\times U(q))\otimes \Lambda_{2pq}.
\end{eqnarray*}

We introduce $2pq$ independent real Grassmann variables, $\{\psi_{jk}^1|1\le j\le q,\, 1\le k\le p\}$ and $\{\psi_{jk}^2|1\le j\le q,\, 1\le k\le p\}$. We add them up into $pq$ complex Grassmann variables labeled as 
\[\{\psi_{jk}=\psi_{jk}^1+i\psi_{jk}^2;j=1,\cdots,q;k=1,\cdots,p\}\]
which leads to the $q\times p$ matrix $\psi$. The real Grassmann algebra generated by $\{\psi_{jk}^a\}$ is denoted by $\Lambda_{2pq}$. Define the $p\times p$ matrix $A$ and the $q\times q$ matrix $B$ with entries in the even part of $\mC\otimes\Lambda_{2pq}$, $\mC\otimes\Lambda_{2pq}^{(ev)}$ as Taylor expansions,
\begin{eqnarray}
\label{defAB2}
A=\sqrt{I_p-i\psi^\dagger \psi}&\mbox{and}&B=\sqrt{I_{q}-i\psi \psi^\dagger}.
\end{eqnarray}
Note that we do not use complex Grassmann variables, just complexified real Grassmann variables, therefore $\overline{\overline{\psi_{jk}}}=\psi_{jk}$ and $\overline{\psi_{jk}\psi_{il}}=\overline{\psi_{jk}}\,\,\overline{\psi_{il}}$. In particular this means that $i\overline{\psi_{jk}}\psi_{jk}=-2\psi^1_{jk}\psi^2_{jk}$ is real.

Similarly to Theorem \ref{Xifo} the following embedding of the bialgebra $\cA$ in the algebra of functions on $\underline{U(p|q)}$ can be calculated, which will be needed to calculate equation \eqref{GenexprInt} for $U(p|q)$.
.
\begin{theorem}
\label{Xifo2}
Consider the matrices $x\in \mC\otimes\left[\cC^\infty(U(p))\right]^{p\times p}$ and $y \in \mC\otimes\left[\cC^\infty(U(q))\right]^{q\times q}$ of matrix elements of the fundamental representation of $U(p)$ and $U(q)$. The matrix
\begin{eqnarray*}
X&=&\left( \begin{array}{cc} X_{I,I}&X_{I,II}\\  \vspace{-3.5mm} \\X_{II,I}&X_{II,II}
\end{array}
 \right)\in\mC\otimes\left[\cC^\infty\left( U(p) \right)\otimes \cC^\infty(U(q))\otimes \Lambda_{2pq}\right]^{(p+q)\times (p+q)}
\end{eqnarray*}
defined by
\begin{eqnarray*}
X_{I,I}=xA & &X_{I,II}=ix\psi^\dagger y\\
X_{II,I}=\psi & &X_{II,II}=By
\end{eqnarray*}
satisfies the relations of $U(p|q)$ in equation \eqref{defrelations2}.
\end{theorem}
\begin{proof}
The exact same techniques as in the proof of Theorem \ref{Xifo} can be used.
\end{proof}

Again the multiplication introduced from this embedding will turn the supermanifold $\underline{U(p|q)}$ into a Lie supergroup with Harish-Chandra pair $(U(p|q),\mathfrak{u}(p|q))$.
\begin{theorem}
\label{defsgr2}
The Lie supergroup with supermanifold $\underline{U(p|q)}$ equipped with the multiplication $\mu:\underline{U(p|q)}\otimes \underline{U(p|q)} \to \underline{U(p|q)}$ and the involutive superdiffeomorphism $\nu:\underline{U(p|q)} \to \underline{U(p|q)}$ given below is the unitary supergroup $U(p|q)$. The multiplication $\mu=(\mu_0,\mu^\sharp)$ is given by
\begin{eqnarray*}
\mu^\sharp(X_{ij})&=&\sum_{k=1}^{m+2n}(-1)^{([i]+[k])([k]+[j])}X_{ik}\otimes X_{kj},
\end{eqnarray*}
for $X_{ij}$ defined in Theorem \ref{Xifo2} with the property 
\[
\mu^\sharp(\overline{f})=\overline{\mu^\sharp(f)} \quad \text{for any $f\in \mC\otimes\cC^\infty\left( U(p) \times U(q)\right)\otimes \Lambda_{2pq}$}.
\]
The involutive superdiffeormorphism $\nu=(\nu_0,\nu^\sharp)$ is defined by
\begin{eqnarray*}
\nu^\sharp(X_{I,I})= X_{I,I}^\dagger & & \nu^\sharp (X_{I,II})=-iX_{II,I}^\dagger\\
\nu^\sharp(X_{II,I})=-iX^\dagger_{I,II} & &\nu^\sharp (X_{II,II})=X_{II,II}^\dagger.
\end{eqnarray*}
\end{theorem}
\begin{proof}
This theorem follows essentially from the fact that the invariant derivations generate $\mathfrak{u}(p|q)$ (see Theorem \ref{liealgmor} below) and the equivalence of categories between Harish-Chandra pairs and Lie supergroups.
\end{proof}

\subsection{Action of the Harish-Chandra pair $(U(p)\times U(q),\mathfrak{u}(p|q))$}

Similarly to Lemma \ref{Liegract} we can calculate the action of $U(p)\times U(q)$ on $\cC^\infty(U(p)\times U(q))\otimes \Lambda_{2pq}$.
\begin{lemma}
\label{Liegract2}
The left co-action of $U(p)\times U(q)$ on $\cO_{U(p|q)}$, $\varphi^\sharp=(\delta^\sharp\otimes id^\sharp)\circ\mu^\sharp$, satisfies
\begin{eqnarray*}
\varphi^\sharp(\psi_{ij})=\sum_{k=1}^{q} y_{ik}\otimes \psi_{kj},&&\varphi^\sharp(\overline\psi_{ij})=\sum_{k=1}^{q} \overline{y}_{ik}\otimes \overline\psi_{kj}\\
\varphi^\sharp(x_{ij})=\sum_{k=1}^px_{ik}\otimes x_{kj}\quad&\mbox{and}& \varphi^\sharp(y_{ij})=\sum_{k=1}^{q} y_{ik}\otimes y_{kj}.
\end{eqnarray*}
\end{lemma}

Each element $P$ of the Lie algebra $\mathfrak{u}(p)$ can be identified with a matrix $P\in\mC^{p\times p}$ satisfying $P^\dagger=-P$. With slight abuse of notation we will also use $P$ for the corresponding invariant real derivation on $U(p)$, which satisfies
\begin{eqnarray*}
P(x_{ab})&=&\sum_{k=1}^pP_{ka}x_{kb}.
\end{eqnarray*}
This is a real derivation in the sense that $\overline{P(f)}=P(\overline{f})$ for $f\in\mC\otimes \cC^\infty(U(p))$. It is straightforward to derive that
\begin{eqnarray*}
P_1(P_2(x_{ab}))-P_2(P_1(x_{ab}))&=&[P_1,P_2](x_{ab}),
\end{eqnarray*}
so the Lie algebra structure is preserved by this assigning of the elements of $\mathfrak{u}(p)$ to real invariant derivations on $U(p)$. We make the same identification between the realization of $\mathfrak{u}(q)$ as anti-hermitian matrices and as invariant derivations on $U(q)$. The embedding of $\mathfrak{u}(p)$ and $\mathfrak{u}(q)$ into $\mathfrak{u}(p)\oplus\mathfrak{u}(q)$ is denoted by $\iota_1$ and $\iota_2$ respectively.

\begin{definition}
\label{liealgmordef}
Let $\tilde{\cdot}:\mathfrak{u}(p|q)\to$Der$\left(\cC^\infty(U(p)\times U(q))\otimes\Lambda_{2pq}\right)$ be a super vector space morphism defined by
\begin{eqnarray*}
\widetilde{D}&=&\left(\delta^\sharp_{0}\circ \iota_1(P)\otimes id^\sharp\right)\circ\mu^\sharp+\left(\delta^\sharp_{0}\circ \iota_2(Q)\otimes id^\sharp\right)\circ\mu^\sharp+\sum_{j=1}^p\sum_{k=1}^q(C_{jk}Y_{kj}+\overline{C}_{jk}\overline{Y}_{kj}) \end{eqnarray*}
for $D=\left( \begin{array}{cc} P&C\\  \vspace{-3.5mm} \\-iC^\dagger&Q\end{array} \right)$. The odd derivations $Y_{kj}$ and $\overline{Y}_{kj}$ are defined as
\begin{eqnarray}
\label{defY}
Y_{kj}=\left(\delta^\sharp_0\circ\partial_{\psi_{kj}}\otimes id^\sharp \right)\circ\mu^\sharp&\mbox{and}&\overline{Y}_{kj}=\left(\delta^\sharp_0\circ\partial_{\overline{\psi}_{kj}}\otimes id^\sharp \right)\circ\mu^\sharp.
\end{eqnarray}
\end{definition}

The derivations $Y_{kj}$ and $\overline{Y}_{kj}$ are not real ($Y_{kj}\in\mC\otimes$Der$\left(\cC^\infty(U(p)\times U(q))\otimes\Lambda_{2pq}\right)$), nevertheless the combinations taken in the morphism $\tilde{\cdot}$ correspond to real derivations.

 This morphism is clearly a bijection between $\mathfrak{u}(p|q)$ and the real invariant derivations on $U(p|q)$. We can prove that it corresponds to a Lie superalgebra morphism.

\begin{theorem}
\label{liealgmor}
The isomorphism of $\mR$-super vector spaces between $\mathfrak{u}(p|q)$ and the real invariant derivations on $U(p|q)$ given in Definition \ref{liealgmordef} is a Lie superalgebra morphism.
\end{theorem}
\begin{proof}
A straightforward calculation shows that the relations
\begin{eqnarray}
\label{eigcomplder}
Y_{ij}(X_{\alpha\beta})=(-1)^{[\beta]}\delta_{i+p,\alpha}X_{j\beta}&\mbox{ and }\overline{Y}_{ij}(X_{\alpha\beta})=-i(-1)^{[\beta]}\delta_{j\alpha}X_{i+p,\beta}.
\end{eqnarray}
hold for the odd invariant derivations. For each matrix $C\in\mC^{p\times q}$ we assign the matrix $\hat{C}\in\mC^{(p+q)\times (p+q)}$, given by\[\hat{C}=\left( \begin{array}{cc} 0&C\\  \vspace{-3.5mm} \\-iC^\dagger&0
\end{array}
 \right).\]
Using this, the result \eqref{eigcomplder} can be rewritten as
\begin{eqnarray*}
\sum_{i=1}^q\sum_{j=1}^p\left(C_{ji}Y_{ij}+\overline{C}_{ji}\overline{Y}_{ij}\right)(X_{\alpha\beta})&=&(-1)^{[\beta]}\sum_{\gamma=1}^{p+q}\hat{C}_{\gamma\alpha}X_{\gamma\beta}.
\end{eqnarray*}
The even derivations satisfy
\begin{eqnarray*}
\widetilde{\iota_1(P)}(X_{\alpha\beta})=\sum_{k=1}^p\left(\iota_1(P)\right)_{k\alpha}X_{k\beta}&\mbox{and}&\widetilde{\iota_2(Q)}(X_{\alpha\beta})=\sum_{l=1}^q\left(\iota_2(Q)\right)_{l+p,\alpha}X_{l+p,\beta}.
\end{eqnarray*}
Putting these results together we obtain for $D\in\mathfrak{u}(p|q)\subset\mC^{(p+q)\times (p+q)}$,
\begin{eqnarray*}
\widetilde{D}(X_{\alpha\beta})&=&\sum_{\gamma=1}^{p+q}(-1)^{[\beta]([\alpha]+[\gamma])}D_{\gamma\alpha}X_{\gamma\beta}
\end{eqnarray*}
The theorem follows from this result.
\end{proof}

All derivations can be expressed in terms of the elements of $\mathfrak{u}(p)\oplus \mathfrak{u}(q)$ and the Grassmann derivatives. The even derivations satisfy
\begin{eqnarray*}
\widetilde{\iota_1(P)}=\iota_1(P)&\mbox{and}&\widetilde{\iota_2(Q)}=\iota_2(Q)+\sum_{ij=1}^q\sum_{k=1}^pQ_{ij}\left(\psi_{ik}\partial_{\psi_{jk}}-\overline{\psi}_{jk}\partial_{\overline{\psi}_{ik}}\right).
\end{eqnarray*}
In order to calculate the expression for the odd derivations we introduce the complex valued invariant derivations $S_{ij}$ on $U(p)$, $1\le i,j\le p$, given by
\begin{eqnarray*}
S_{ij}(x_{ab})&=&\delta_{ja}x_{ib}.
\end{eqnarray*}
These derivations satisfy $\overline{S_{ij}}=-S_{ji}$. The derivations $P\in\mathfrak{u}(p)$ correspond to $P=\sum_{ij=1}^pP_{ij}S_{ij}$, which again shows that $\overline{P}=P$. The corresponding (complex) invariant derivations on $U(q)$ are denoted by $T_{ij}$. The derivations $S_{ij}$ and $T_{ij}$ are non-real linear combinations of elements of $\mathfrak{u}(p)$ and $\mathfrak{u}(q)$ and form a basis for the $\mC$-vector space of complex invariant derivations on $U(p)$ and $U(q)$. In fact these spaces correspond to $\mathfrak{gl}(p;\mC)$ and $\mathfrak{gl}(q;\mC)$.
\begin{lemma}
\label{expressionY1}
For $1\le i\le q$ and $1\le j\le p$, the invariant derivations $Y_{ij}$ satisfy the relation
\begin{eqnarray*}
Y_{ij}&=&\sum_{k=1}^p(xA)_{jk}\partial_{\psi_{ik}}+\sum_{k,l=1}^{p}f_{ij}^{kl}S_{kl}+\sum_{s,t=1}^qg^{st}_{ij}T_{st}
\end{eqnarray*}
with
\begin{eqnarray*}
f_{ij}^{kl}&=&-\sum_{a,b=1}^p(xA)_{ja}\left(\partial_{\psi_{ia}}(xA)_{kb}\right) (A^{-1}x^{-1})_{bl},\\
g_{ij}^{st}&=&-\sum_{a=1}^p\sum_{r=1}^q(xA)_{ja}\left(\partial_{\psi_{ia}}(B_{rt})\right)B^{-1}_{sr}-iB^{-1}_{si}(x\psi^\dagger)_{jt},
\end{eqnarray*}
with $A$ and $B$ defined in equation \ref{defAB2}.
\end{lemma}
\begin{proof}
Since $Y_{ij}$ is a derivation on $U(p|q)$ it has to be of the form
\begin{eqnarray*}
\sum_{k=1}^p\sum_{s=1}^qh_{ij}^{sk}\partial_{\psi_{sk}}+\sum_{k,l=1}^{p}f_{ij}^{kl}S_{kl}+\sum_{s,t=1}^qg^{st}_{ij}T_{st}+\sum_{k=1}^p\sum_{s=1}^qk_{ij}^{sk}\partial_{\overline{\psi}_{sk}}
\end{eqnarray*}
with $h_{ij}^{sk}$, $f_{ij}^{kl}$, $g^{st}_{ij}$ and $k_{ij}^{sk}$ elements of $\mC\otimes\cC^\infty(U(p)\times U(q))\otimes \Lambda_{2pq}$. Since
\begin{eqnarray*}
\mu^\sharp\left(\overline{\psi}_{sk}\right)=\overline{\mu^\sharp\left(\psi_{sk}\right)}&=&\sum_{j=1}^p\overline{\psi}_{sj}\otimes \overline{(xA)_{jk}}+\sum_{t=1}^q\overline{(By)_{st}}\otimes\overline{\psi}_{tk},
\end{eqnarray*}
we find that $Y_{ij}(\overline{\psi}_{sk})=0$, which implies $k_{ij}^{sk}=0$. Equation \eqref{eigcomplder} implies that $Y_{ij}(\psi_{sk})=\delta_{is}(xA)_{jk}$, therefore $h_{ij}^{sk}=\delta_{is}(xA)_{jk}$ holds. Equation \eqref{eigcomplder} also implies $Y_{ij}((xA)_{ab})=0$ holds, which leads to the relation
\begin{eqnarray*}
\sum_{k=1}^p(xA)_{jk}\left(\partial_{\psi_{ik}}(xA)_{ab}\right)+\sum_{l=1}^{p}f_{ij}^{al}(xA)_{lb}&=&0.
\end{eqnarray*}
This defines $f^{kl}_{ij}$. Finally equation \eqref{eigcomplder} implies $Y_{ij}\left((By)_{ab}\right)=-i\delta_{ia}(x\psi^\dagger y)_{jb}$, therefore to the equation
\begin{eqnarray*}
\sum_{k=1}^p(xA)_{jk}\partial_{\psi_{ik}}\left((By)_{ab}\right)+\sum_{s,t=1}^qg^{st}_{ij}B_{as}y_{tb}&=&-i\delta_{ia}(x\psi^\dagger y)_{jb}
\end{eqnarray*}
must hold, which defines $g^{st}_{ij}$.
\end{proof}

This lemma also implies the expression for the derivation $\overline{Y}_{ij}$ by complex conjugation. The corresponding real invariant derivations are given by $Y_{ij}+\overline{Y}_{ij}$ and $i(Y_{ij}-\overline{Y}_{ij})$.

\subsection{Invariant integration on $U(p|q)$}

The following corollary is essential to calculate the invariant integral on $U(p|q)$.
\begin{corollary}
\label{expressionY2}
The invariant derivation $Y_{ij}$ from equation \eqref{defY} and Lemma \ref{expressionY1} satisfies the relation
\begin{eqnarray*}
Y_{ij}&=&\sum_{k=1}^p\partial_{\psi_{it}}(xA)_{jt}+\sum_{k,l=1}^{p}S_{kl}f_{ij}^{kl}+\sum_{s,t=1}^qT_{st}g^{st}_{ij}
\end{eqnarray*}
with the functions $f_{ij}^{kl}$ and $g^{st}_{ij}$ as defined in Lemma \ref{expressionY1}.
\end{corollary}
\begin{proof}
This corollary follows from the calculation
\begin{eqnarray*}
\sum_{k,l=1}^pS_{kl}(f_{ij}^{kl})&=&-\sum_{k=1}^p\partial_{\psi_{ik}}\left((xA)_{jk}\right)
\end{eqnarray*}
and the fact that $T_{st}$ and $g_{ij}^{st}$ commute.
\end{proof}

We define the Berezin integral on $\Lambda_{2pq}$ as
\begin{eqnarray*}
\int_{B_{2pq}}&=&\prod_{i=1}^q\prod_{j=1}^p\left(\partial_{\psi_{ij}^1}\partial_{\psi_{ij}^2}\right).
\end{eqnarray*}

The invariant integral on the supergroup $U(p|q)$ can now be calculated.
\begin{theorem}\label{InvInt2}
The unique invariant integral on $U(p|q)$,
\begin{eqnarray*}
\int_{U(p|q)}:\cO_{U(p|q)}(U(p)\times U(q))=\cC^\infty(U(p)\times U(q))\otimes \Lambda_{2pq}\to\mR
\end{eqnarray*}
is given by
\begin{eqnarray*}
\int_{U(p|q)}\cdot&=&\int_{U(p)\times U(q)}\int_{B_{2pq}}\cdot.
\end{eqnarray*}
\end{theorem}
\begin{proof}
The expression is clearly $U(p)\times U(q)$-invariant. The fact that
\begin{eqnarray*}
\int_{U(p|q)}\circ X&=&0
\end{eqnarray*}
for every $X\in\mathfrak{u}(p|q)_1$ follows immediately from Corollary \ref{expressionY2}.
\end{proof}

This expression immediately yields $\int_{U(p|q)}1=0$ if $pq>0$, a manifestation of the fact that there are non-semisimple finite dimensional $\mathfrak{u}(p|q)$-modules. Using the same arguments given in the proof of Corollary \ref{Int-nondegeneracy}, we can prove the following result.
\begin{corollary}\label{Int-nondegeneracy2}
The invariant integral on $U(p|q)$ is non-degenerate in the sense that the supersymmetric bilinear form
\begin{eqnarray*}
\left(f,g\right)&=&\int_{U(p|q)}fg
\end{eqnarray*}
on $\cO_{U(p|q)}(U(p)\times U(q))$ is non-degenerate.
\end{corollary}

According to the super adjoint in Section \ref{defalg}, we define
\begin{eqnarray*}
X^\ast=\left( \begin{array}{cc} X_{I,I}^\ast&X_{II,I}^\ast\\  \vspace{-3.5mm} \\X_{I,II}^\ast&X_{II,II}^\ast\end{array} \right)&=&\left( \begin{array}{cc} X_{I,I}^\dagger&iX_{II,I}^\dagger\\  \vspace{-3.5mm} \\iX_{I,II}^\dagger&X_{II,II}^\dagger\end{array} \right).
\end{eqnarray*}
Using this definition, the relations on the matrix $X$ \eqref{defrelations2} can be rewritten as
\begin{eqnarray*}
X_{I,I}^\ast X_{I,I}+X_{II,I}^\ast X_{II,I}&=&I_p,\\
-X_{I,I}^\ast X_{I,II}+X_{II,I}^\ast X_{II,II}&=&0,\\
X_{I,II}^\ast X_{I,II}+X_{II,II}^\ast X_{II,II}&=& I_q.
\end{eqnarray*}
This means that the super algebra generated by $X_{ij}$ and $(X^\ast)_{kl}$, $\cB=$Alg$(\{X_{ij}\},\{(X^\ast)_{kl}\})$ is a real algebra. The expression in Theorem \ref{InvInt2}, shows that the integral evaluated on $\cB$ will give real values. As in equation \eqref{GenexprInt}, the formula in Theorem \ref{InvInt2} should be interpreted as
\begin{eqnarray*}
\int_{U(p|q)}f(X,X^\ast)&=&\int_{U(p)\times U(q),x\times y}\int_{B_{2pq},\psi}f(X(x,y,\psi),X^\ast(x,y,\psi)).
\end{eqnarray*}

As an illustration we calculate the case $U(1|1)$. The trivial classical result is given by $\int_{U(1),x}x^k\overline{x}^l=\delta_{kl}$. Therefore the integral
\[\int_{U(1|1)}X_{11}^{\alpha_{11}}X_{12}^{\alpha_{12}}X_{21}^{\alpha_{21}}X_{22}^{\alpha_{22}}({X}^\ast)_{11}^{\beta_{11}}({X}^\ast)_{12}^{\beta_{12}}({X}^\ast)_{21}^{\beta_{21}}({X}^\ast)_{22}^{\beta_{22}}\]
will be zero unless $\alpha_{11}+\alpha_{12}=\beta_{11}+\beta_{21}$ and $\alpha_{22}+\alpha_{12}=\beta_{22}+\beta_{21}$ hold. Because of the Berezin integral, $\psi$ and $\overline{\psi}$ must appear an equal amount of times, so $\alpha_{21}+\beta_{21}=\alpha_{12}+\beta_{12}$ needs to hold in order for the integral to be different from zero.

If we assume first $\alpha_{21}+\beta_{21}=\alpha_{12}+\beta_{12}=0$ (therefore $\alpha_{11}=\beta_{11}$ and $\alpha_{22}=\beta_{22}$), the integral reduces to
\begin{eqnarray*}
\int_{B_2}\sqrt{1-i\overline{\psi}\psi}^{\alpha_{11}+\beta_{11}}\sqrt{1+i\overline{\psi}\psi}^{\alpha_{22}+\beta_{22}}&=&\int_{B_2}(1-\frac{1}{2}i\overline{\psi}\psi)^{2\alpha_{11}}(1+\frac{1}{2}i\overline{\psi}\psi)^{2\alpha_{22}}\\
&=&\int_{B_2}(1+2\alpha_{11}\psi^2\psi^1)(1-2\alpha_{22}\psi^2\psi^1)\\
&=&2(\alpha_{11}-\alpha_{22}).
\end{eqnarray*}

The other case which gives a nonzero result is $\alpha_{21}+\beta_{21}=\alpha_{12}+\beta_{12}=1$ (with $\alpha_{11}+\alpha_{12}=\beta_{11}+\beta_{21}$ and $\alpha_{22}+\alpha_{12}=\beta_{22}+\beta_{21}$). The integral then reduces to
\begin{eqnarray*}
&&\int_{U(1|1)}x^{\alpha_{11}}(ix\overline{\psi}y)^{\alpha_{12}}(\psi)^{\alpha_{21}}y^{\alpha_{22}}\overline{x}^{\beta_{11}}(i\overline{\psi})^{\beta_{12}}(\overline{x}\psi\overline{y})^{\beta_{21}}\overline{y}^{\beta_{22}}\\
&=&i^{\alpha_{12}+\beta_{21}}\int_{B_2}(\overline{\psi})^{\alpha_{12}}(\psi)^{\alpha_{21}}(\overline{\psi})^{\beta_{12}}(\psi)^{\beta_{21}}=2(-1)^{\alpha_{21}\beta_{12}}.
\end{eqnarray*}

Summarizing, we obtained the result
\begin{eqnarray*}
&&\int_{U(1|1)}X_{11}^{\alpha_{11}}X_{12}^{\alpha_{12}}X_{21}^{\alpha_{21}}X_{22}^{\alpha_{22}}({X}^\ast)_{11}^{\beta_{11}}({X}^\ast)_{12}^{\beta_{12}}({X}^\ast)_{21}^{\beta_{21}}({X}^\ast)_{22}^{\beta_{22}}\\
&=&\delta_{\alpha_{12}+\alpha_{21}+\beta_{12}+\beta_{21},0}\,\delta_{\alpha_{11},\beta_{11}}\delta_{\alpha_{22},\beta_{22}}\,2(\alpha_{11}-\alpha_{22})\\
&+&\delta_{\alpha_{12}+\beta_{12},1}\, \delta_{\alpha_{21}+\beta_{21},1}\, \delta_{\alpha_{11}+\alpha_{12},\beta_{11}+\beta_{21}}\delta_{\alpha_{12}+\alpha_{22},\beta_{21}+\beta_{22}}\,2(-1)^{\alpha_{21}\beta_{12}}.
\end{eqnarray*}

\section{The superspace $UOSp(m|2n)$}\label{sect:UOSp}

In this section we extend our results to the more physical approach to supermanifolds. We will be very brief since this is just a reformulation of the results into another language. We use complex Grassmann variables, with the second kind of complex conjugation. This implies that for two complex Grassmann variables $\alpha$ and $\beta$, the following conjugation rules hold:
\begin{eqnarray*}
\overline{\overline{\alpha}}=-\alpha&\mbox{and}&\overline{\alpha\beta}=\overline{\alpha}\,\overline{\beta}.
\end{eqnarray*}
This implies that $\overline{(\alpha\overline{\alpha})}=\alpha\overline{\alpha}$, thus $\alpha\overline{\alpha}$ is real, and $\overline{\overline{\alpha\beta}}=\alpha\beta$. We assume a Grassmann algebra $\Lambda_{2Q}$ generated by $Q$ such complex Grassmann variables and their complex conjugates with $Q$ very large or infinity. Another way to define supergroups (see also \cite{MR2081650, MR1124825, MR0725288}) is as follows. The unitary group $U(p|q)$ consists of the matrices $X\in(\Lambda_{2Q})^{(p+q)\times (p+q)}$ (where the $p\times p$ and $q\times q$ block have entries from the even part of $\Lambda_Q$ and the $p\times q$ and $q\times p$ block from the odd part) that satisfy
\begin{eqnarray}
\label{defrelations3}
\left( \begin{array}{cc} X_{I,I}^\dagger&X_{II,I}^\dagger\\  \vspace{-3.5mm} \\X_{I,II}^\dagger&X_{II,II}^\dagger\end{array} \right)\left( \begin{array}{cc} X_{I,I}&-X_{I,II}\\  \vspace{-3.5mm} \\X_{II,I}&X_{II,II}\end{array} \right)&=&I_{p+q}.
\end{eqnarray}
Usually the minus sign appears with $X_{I,II}^\dagger$ rather than with $X_{I,II}$ but that can be obtained from a simple relabeling. Then we can introduce coordinates as follows. We introduce $q\times p$ independent complex Grassmann variables in the matrix $\psi$. These can be considered as $q\times p$ arbitrary odd elements of $\Lambda_{2Q}$. Then we find that
\begin{eqnarray*}
X_{I,I}=x\widetilde{A} & &X_{I,II}=x\psi^\dagger y\\
X_{II,I}=\psi & &X_{II,II}=\widetilde{B}y
\end{eqnarray*}
satisfy the relations \eqref{defrelations3}, with again $x$, $y$ the matrix elements of $U(p)$, $U(q)$, but now $\widetilde{A}=\sqrt{I_p-\psi^\dagger\psi}$ and $\widetilde{B}=\sqrt{I_q-\psi\psi^\dagger}$. The invariance of the integral is then expressed as
\begin{eqnarray*}
\int f(AX)=\int f(X)
\end{eqnarray*}
for all $A\in U(p|q)$. This can be rewritten in a Hopf-algebraic way. Very similarly as in our main approach to supergroups it is then proved that the integration of the form
\begin{eqnarray*}
\int_{U(p|q)}f(X,X^\dagger)&=&\int_{U(p)\times U(q),x\times y}\int_{B_{2pq},\psi}f(X(x,y,\psi), X^\dagger(x,y,\psi))
\end{eqnarray*}
is the invariant integration.

The superspace $UOSp(m|2n)$ (which is not actually a supergroup, see e.g. discussion in \cite{MR2081650} and \cite{MR2276521}) is then defined as $OSp(m|2n;\mC)\cap U(m|2n)$, so the complex matrix elements $X$ satisfy both equations \eqref{defrelations} and \eqref{defrelations3}. To obtain coordinates, we introduce $2mn$ complex Grassmann variables $\theta$ which satisfy $\overline{\theta}=-J\theta$. This leads to $mn$ independent complex Grassmann variables. The identification
\begin{eqnarray*}
X_{I,I}=xA & &X_{I,II}=x\theta^T J z\\
X_{II,I}=\theta & &X_{II,II}=Bz
\end{eqnarray*}
with $x$ the matrix elements of $O(m)$, $z$ the matrix elements of $USp(2n)$, $A=\sqrt{I_m-\theta^TJ\theta}=\sqrt{I_m-\theta^\dagger\theta}$ and $B=\sqrt{I_m-\theta\theta^TJ}=\sqrt{I_m-\theta\theta^\dagger}$, satisfies equations \eqref{defrelations} and \eqref{defrelations3}. The invariant integral on $UOSp(m|2n)$ then takes the form
\begin{eqnarray*}
\int_{UOSp(m|2n)}f(X)&=&\int_{O(m)\times USp(2n),x\times z}\int_{B_{2mn},\theta}\det(I_m-\theta^\dagger\theta)^{-1/2}\, f(X(x,z,\theta)).
\end{eqnarray*}
This is an explicit construction of the invariant measure on $UOSp(m|2n)$ implicitly used in \cite{MR2081650}.

\section{Conclusion}
We conclude the paper with a brief summary of the main results. By treating a Lie supergroup as a Harish-Chandra pair consisting of a Lie superalgebra and an ordinary Lie group, we have developed explicit formulae for the invariant integrals on the orthosymplectic Lie supergroup $OSp(m|2n)$ and unitary Lie supergroup $U(p|q)$ for all $m, n$ and $p, q$. The results are presented in Theorem \ref{InvInt} and Theorem \ref{InvInt2}.  The formulae are elegant and simple to use. For example, by using the formulae, we have proved the non-degeneracy of the integrals in Corollary \ref{Int-nondegeneracy} and Corollary \ref{Int-nondegeneracy2} with very little effort. This is a problem which eluded solution before.

In Section \ref{sect:UOSp}, we have obtained from the invariant integral on $OSp(m|2n)$ an explicit construction of the invariant measure on $UOSp(m|2n)$ implicitly used in \cite{MR2081650}. This already demonstrates the relevance of our work to the theory of random super matrices, and we expect further applications of results in the present paper to lead to significant progress in the area.

\end{document}